\documentclass[onecolumn,a4paper,accepted=2019-04-01]{quantumarticle}

\pdfoutput=1

\usepackage[utf8]{inputenc}
\usepackage[T1]{fontenc}
\usepackage{amsmath}
\usepackage{amssymb}
\usepackage{amsthm}
\usepackage{mathtools}
\usepackage{bbm}
\usepackage{braket}
\usepackage{qcircuit}
\usepackage{hyperref}
\usepackage[backend=bibtex8,bibencoding=ascii,%
    language=auto,%
    style=numeric-comp,%
    doi=true,
    isbn=false,
    url=false,
    maxbibnames=10, 
    backref=true,%
    natbib=true 
    ]{biblatex}
\usepackage{doi}
\usepackage{tikz}
\usepackage{tikz-3dplot}
\usepackage{algorithm}
\usepackage{algpseudocode}
\usepackage{subcaption}
\usepackage{thmtools, thm-restate}

\newcommand{\R}{\mathbb{R}}

\newcommand{\C}{\mathbb{C}}
\newcommand{\Z}{\mathbb{Z}}

\newcommand{\F}{\mathbb{F}}

\DeclareMathOperator{\tr}{tr}
\DeclareMathOperator{\sgn}{sgn}

\DeclareMathOperator{\conv}{conv}

\DeclareMathOperator{\End}{End}
\DeclareMathOperator{\Aut}{Aut}

\DeclareMathOperator{\diag}{diag}
\DeclareMathOperator{\Sym}{Sym}
\DeclareMathOperator{\wt}{wt}
\DeclareMathOperator{\Ad}{Ad}
\DeclareMathOperator{\poly}{poly}

\newcommand{\one}{\mathbbm{1}}

\newcommand{\GL}{\mathrm{GL}}
\newcommand{\Gl}{\GL}

\newcommand{\Sp}{\mathrm{Sp}}

\newcommand{\Po}{\mathcal{P}}
\newcommand{\stab}{\mathrm{stab}}
\newcommand{\SP}{\mathrm{SP}}
\newcommand{\ASP}{\mathrm{Q}}
\newcommand{\graph}{\mathrm{graph}}

\newcommand{\Cl}{\mathcal{C}}

\newcommand{\ECl}{\mathrm{E}\mathcal{C}}

\newcommand{\Pa}{\mathcal{P}}

\newcommand{\G}{G}

\newcommand{\RM}{\mathcal{R}}

\theoremstyle{plain}
\newtheorem{theorem}{Theorem}
\newtheorem{proposition}{Proposition}

\newtheorem{lemma}{Lemma}

\theoremstyle{definition}

\newtheorem{problem}{Problem}

\theoremstyle{remark}

\newcommand{\ie}{i.\,e.}

\newcommand{\eg}{e.\,g.}

\usetikzlibrary{shapes.misc}
\usetikzlibrary{arrows}
\usetikzlibrary{decorations.text}
\usetikzlibrary{shapes.symbols}
\usetikzlibrary{shapes.geometric}
\usetikzlibrary{decorations.pathreplacing}

\newsavebox{\twosubbox}

\title{Robustness of Magic and\newline Symmetries of the Stabiliser Polytope}
\author{Markus Heinrich}
\email{markus.heinrich@uni-koeln.de}
\orcid{0000-0002-1334-7885}
\author{David Gross}
\affiliation{%
  Institute for Theoretical Physics, University of Cologne, 50937 Cologne, Germany
}%

\begin{document}

\abstract{
We give a new algorithm for computing the \emph{robustness of magic}---a measure of the utility of quantum states as a computational resource.
Our work is motivated by the \emph{magic state model} of fault-tolerant quantum computation. 
In this model, all unitaries belong to the Clifford group. 
Non-Clifford operations are effected by injecting non-stabiliser states, which are referred to as \emph{magic states} in this context.
The \emph{robustness of magic} measures the complexity of simulating such a circuit using a classical Monte Carlo algorithm.
It is closely related to the degree negativity that slows down Monte Carlo simulations through the infamous \emph{sign problem}.
Surprisingly, the robustness of magic is \emph{sub}multiplicative.
This implies that the classical simulation overhead scales subexponentially with the number of injected magic states---better than a naive analysis would suggest.
However, determining the robustness of $n$ copies of a magic state is difficult, as its definition involves a convex optimisation problem in a $4^n$-dimensional space.
In this paper, we make use of inherent symmetries to reduce the problem to $n$ dimensions.
The total run-time of our algorithm, while still exponential in $n$, is super-polynomially faster than previously published methods.
We provide a computer implementation and give the robustness of up to 10 copies of the most commonly used magic states.
Guided by the exact results, we find a finite hierarchy of approximate solutions where each level can be evaluated in polynomial time and yields rigorous upper bounds to the robustness. 
Technically, we use symmetries of the stabiliser polytope to connect the robustness of magic to the geometry of a low-dimensional convex polytope generated by certain \emph{signed quantum weight enumerators}.
As a by-product, we characterised the automorphism group of the stabiliser polytope, and, more generally, of projections onto complex projective 3-designs.
}


\section{Introduction}
\label{sec:intro}

In fault-tolerant quantum computation (for a recent review, see Ref.~\cite{campbell_roads_2017}), each logical qubit is encoded in a non-local subspace of a number of physical qubits.
There are several ways of effecting a unitary transformation of logical qubits.
In the simplest case, logical unitaries can be implemented \emph{transversally}, i.e.\ by local gates acting on the physical qubits.
Unfortunately, a no-go theorem by \textcite{eastin_restrictions_2009} states that there are no quantum codes that allow for a \emph{universal} set of transversal gates.

In the \emph{magic state model} \cite{bravyi_universal_2005}, the logical gate set is chosen to be the Clifford group, which can be implemented transversally in various quantum codes using their physical counterparts.
Any logical non-Clifford gate would promote the Clifford group to universality.
This remaining problem is solved by providing an auxiliary qubit in a non-stabiliser state.
Using a circuit gadget (which only requires Clifford operations), one can turn this auxiliary state into a non-Clifford gate (Fig.~\ref{fig:state_inj}).
The auxiliary qubit state is consumed in the process, so that one such input needs to be injected for each non-Clifford gate.
These inputs are the \emph{magic states} from which the protocol derives its name.

\begin{figure}
 \centering
 \includegraphics[width=0.75\linewidth]{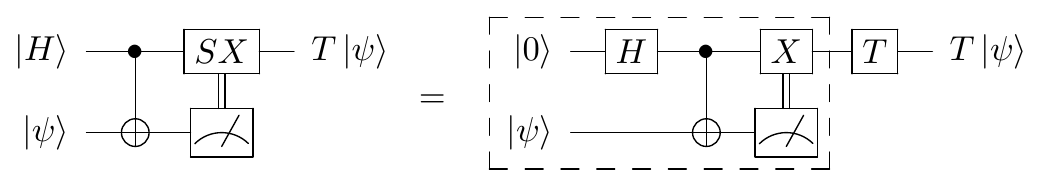}
 \caption{Use of magic state injection to perform a $T$ gate on some input state $\ket\psi$. The state injection circuit can be rewritten as a swap circuit followed by $T$ gate.}
 \label{fig:state_inj}
\end{figure}

A common choice for a non-Clifford gate is the $T$-gate $T=\diag(1,e^{i\pi/4})$, which is realised by the following magic state
\begin{equation}
 \label{H_state}
 \ket{H} := T\ket{+} = \frac{1}{\sqrt{2}} \left (\ket 0 + e^{i\pi/4} \ket 1 \right).
\end{equation}
Moreover, there is a second magic state, $\ket{T}$, which realises the non-Clifford gate $\diag(1,e^{i\pi/6})$. Their Bloch representation is shown in Fig.~\ref{fig:stab_polytope}. 
Interestingly, it has been found that even certain mixed states can ``supply the magic'' to promote a Clifford circuit to universality.
Indeed, a process called \emph{magic state distillation} (Fig.~\ref{fig:magic_state_destillation}) can turn many copies of some mixed state $\rho$ into a pure magic state using Clifford unitaries and computational basis measurements \cite{bravyi_universal_2005,reichardt_quantum_2006}.

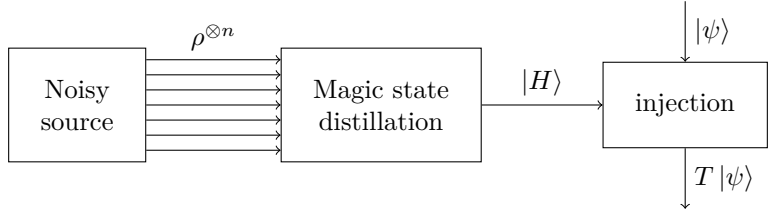
\begin{figure}
\centering
 \begin{tikzpicture}
 \node[draw,inner sep=12pt,align=center] (src) at (0,0) {Noisy \\ source};
 \node[draw,align=center,inner sep=12pt] (dest) at (4,0) {Magic state\\ distillation};
 \node[draw,align=center,inner sep=12pt] (gate) at (8,0) {injection};
 
 \node (in) at (8,1.5) {};
 \node (out) at (8,-1.5) {};
 

 \foreach \i in {-3,...,2}{%
  \draw[->] ([yshift=\i*0.2cm]src.east) to ([yshift=\i*0.2cm]dest.west);
 }
 \draw[->] ([yshift=3*0.2cm]src.east) to node[above,midway] {$\rho^{\otimes n}$} ([yshift=3*0.2cm]dest.west);


 \draw[->] (dest) to node[above,midway] {$\ket{H}$} (gate);
 
 \draw[->] (in) to node[right,midway] {$\ket\psi$} (gate);
 \draw[->] (gate) to node[right,midway] {$T\ket\psi$} (out);
 
\end{tikzpicture}
\caption{Magic state distillation turns a supply of mixed states $\rho$ into a pure magic state, e.g.\ $\ket{H}$ using only Clifford operations.}
\label{fig:magic_state_destillation}
\end{figure}

Magic state distillation motivates the search for quantitative measures of the ``computational utility'' of auxiliary states.
This analysis turns out to be slightly simpler for quantum systems with odd-dimensional Hilbert spaces \cite{veitch_negative_2012,mari_positive_2012,veitch_resource_2014}, as the theory of stabiliser states is somewhat better-behaved in this case, and there is a better-developed toolbox of ``phase space methods'' available in this case (see e.g.~Refs.~\cite{gross_hudsons_2006,zhu2016permutation,karanjai2018contextuality}).
However, as qubits are the paradigmatic systems for quantum computation, quantitative resource theories for multi-qubit magic states have since been developed \cite{howard_application_2017,bravyi_trading_2016}.

The starting point of these theories is the Gottesman-Knill Theorem \cite{nielsen_quantum_2011}. It states that quantum circuits consisting only of preparations of stabiliser states, Clifford unitaries, and computational basis measurements can be efficiently simulated on a classical computer.
Therefore, if the auxiliary states are stabilisers, there can be no quantum computational advantage.
Next, assume that an auxiliary $n$-qubit state $\rho$ is an element of the \emph{stabiliser polytope} $\SP_n$, i.e.\ 
\begin{equation*}
	\rho = \sum_i p_i s_i,
\end{equation*}
where $(p_i)_i$ is a probability distribution and the $s_i=\ket{\psi_i}\bra{\psi_i}$ are stabiliser states.
This readily gives rise to an efficient classical randomised algorithm that will draw outcomes from the same distribution as a quantum computer would \cite{virmani_classical_2005}, provided that one can sample efficiently from the probability distribution $(p_i)_i$:
Indeed, draw $s_i$ with probability $p_i$, and then continue to simulate the further time evolution using Gottesman-Knill.
Thus, density matrices contained in the convex hull of stabiliser states are equally useless as computational resource states in the magic state model (Fig.~\ref{fig:stab_polytope_intro}).

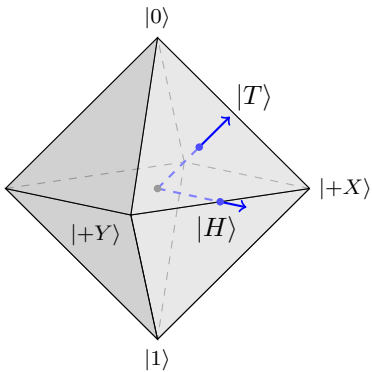
\begin{figure}[h]
\begin{minipage}[c]{0.35\textwidth}
 \begin{tikzpicture}[z=-5,scale=2]
  \coordinate (X1) at (1,0,0);
  \coordinate (X2) at (-1,0,0);
  \coordinate (Z1) at (0,1,0);
  \coordinate (Z2) at (0,-1,0);
  \coordinate (Y1) at (0,0,1);
  \coordinate (Y2) at (0,0,-1);
  \coordinate (H) at (0.71,0,0.71);
  \coordinate (T) at (0.5774,0.5774,0.5774);

  \draw (X1) -- (Y1);
  \draw (X2) -- (Y1);

  \begin{scope}[dashed,opacity=0.6]
  \draw (Y2) -- (Z1);
  \draw (Y2) -- (Z2);
  \draw (X2) -- (Y2);
  \draw (X1) -- (Y2);
  \end{scope}

  \draw (X1) -- (Z1);
  \draw (X2) -- (Z1);
  \draw (Y1) -- (Z1);

  \draw [fill=black!15!white,opacity=0.6] (X1) -- (Z2) -- (Y1) -- (X1);
  \draw [fill=black!30!white,opacity=0.6] (Y1) -- (X2) -- (Z2) -- (Y1);
  \draw [fill=black!15!white,opacity=0.6] (X1) -- (Z1) -- (Y1) -- (X1);
  \draw [fill=black!30!white,opacity=0.6] (Y1) -- (X2) -- (Z1) -- (Y1);

  \draw (X2) -- (Z2);
  \draw (Y1) -- (Z2);
  \draw (X1) -- (Z2);
  
  \draw [dashed,blue!50!white,thick] (0,0,0) -- (0.5,0,0.5);
  \draw [->,blue,thick] (0.5,0,0.5) -- (H);
  \node [below left=3pt,inner sep=0pt] at (H) {$\ket H$};
  \node [circle, fill=blue!70!white, inner sep=1pt] at (0.5,0,0.5) {}; 
  
  \draw [dashed,blue!50!white,thick] (0,0,0) -- (0.333,0.333,0.333);
  \draw [->,blue,thick] (0.333,0.333,0.333) -- (T);
  \node [above right=1pt,inner sep=1pt] at (T) {$\ket T$};
  \node [circle, fill=blue!70!white, inner sep=1pt] at (0.333,0.333,0.333) {};  
  
  \node [circle, fill=gray!80!white, inner sep=1pt] at (0,0,0) {};
  
  \node [above] at (Z1) {\footnotesize $\ket{0}$};
  \node [below] at (Z2) {\footnotesize $\ket{1}$};
  \node [right] at (X1) {\footnotesize $\ket{+X}$};
  \node [below left] at (Y1) {\footnotesize $\ket{+Y}$};

  \end{tikzpicture}
\end{minipage}
\begin{minipage}[c]{0.6\textwidth}
 \caption{Bloch representation of the the two most commonly considered magic states $\ket{H}$ and $\ket{T}$. 
 These states lie outside of the octahedron spanned by 1-qubit stabiliser states having a Bloch vector orthogonal to an edge ($\ket H$) or a facet ($\ket T$) of the stabiliser octahedron. The intersection of their Bloch vector with the facet or edge is marked with a blue dot.
 Certain mixed states can be used to distil these pure states using Clifford unitaries and measurements. However, states lying inside the stabiliser polytope are useless as a resource state.}
 \label{fig:stab_polytope_intro}
\end{minipage}
\end{figure}

Since the stabiliser states $\{s_i\}_i$ span the space of Hermitian operators, any auxiliary state can be expanded as $\rho = \sum_i x_i s_i$, with coefficients $x_i$ that are not necessarily non-negative. However, taking traces on both sides shows that the expansion is \emph{affine}, i.e.\ $\sum_i x_i = 1$.
It is well-known in the theory of Quantum Monte Carlo methods \cite{gubernatis2016quantum} that the probabilistic algorithm sketched above can be extended to the more general scenario. 
However, the runtime will increase with the total amount of ``negativity'' in the expansion coefficients $x_i$.
This is the dreaded \emph{sign problem}.
A precise theory of the simulation runtime in the context of quantum computation has been developed in Ref.~\cite{pashayan_estimating_2015} and applied to the magic state model in Ref.~\cite{howard_application_2017}.
More precisely, they define the \emph{robustness of magic} (RoM) as
\begin{equation}
\label{eq:l1_robustness}
 \RM(\rho) := \min \left\{ \|x\|_1 \; \bigg| \; x\in\R^N: \; \rho = \sum_{i=1}^N x_i s_i \right\},
\end{equation}
where the sum ranges over stabiliser states $\{s_1,\dots,s_N\}$ and  the $\ell_1$-norm
\begin{equation*}
	\|x\|_1 = \sum_{i=1}^N |x_i| = 1+ 2\sum_{i:\, x_i \leq 0} |x_i|
\end{equation*}
measures the ``amount of negativity'' in the affine combination. Then, the number of samples which have to be taken in the Monte Carlo simulation scales as $O(\mathcal{R}(\rho)^2)$ \cite{howard_application_2017,pashayan_estimating_2015}.

In addition to measuring the ``computational utility'' in the above precise sense, the RoM has further interpretations.
For example, it can be used to systematically lower-bound the number of non-Clifford gates required to synthesise certain unitaries, namely those that allow for a magic state realisation \cite{howard_application_2017}.
Lastly, the RoM derives its name from the fact that it quantifies the robustness of a state's computational utility against noise processes.
A precise account of this point of view is given  in Section~\ref{sec:rom}.

Interestingly, the RoM is \emph{sub}multiplicative, i.e.\ $\mathcal{R}(\rho^{\otimes 2}) \leq \mathcal{R}(\rho)^2$, where the inequality is usually strict \cite{howard_application_2017}.
That means that the simulation effort of a magic state circuit grows subexponentially with the number of injected magic states---an intriguing phenomenon.
Therefore, a quantity of interest is the \emph{regularised RoM}:
\begin{equation*}
	\RM_{\mathrm{reg}}(\rho):= \lim_{n\to\infty} \RM(\rho^{\otimes n})^{1/n}.
\end{equation*}
Unfortunately, computing $\mathcal{R}(\rho^{\otimes n})$ seems to be a difficult task. For $\rho$ being a single-qubit state, the tensor power $\rho^{\otimes n}$ lives in an $4^n$-dimensional space, and the sum over the $s_i$ in the definition (\ref{eq:l1_robustness}) of the RoM has to range over the $2^{O(n^2)}$ stabiliser states defined for $n$-qubit systems.
Any direct implementation of the optimisation problem (\ref{eq:l1_robustness}) will thus quickly became computationally intractable---and, indeed, \textcite{howard_application_2017} could carry it out only up to $n=5$.

The starting point of this work is the observation that there is a large symmetry group shared by $\rho^{\otimes n}$ and the stabiliser polytope. Thus, we formulate the optimisation in a space where the joint symmetries have been ``modded out''. 
The space of operators invariant under the joint symmetry group turns out to have a dimension mildly polynomial in $n$. For the especially interesting cases where the state is $\ket H^{\otimes n}$ or $\ket T^{\otimes n}$, the dimension reduces further to exactly $n$.
While the projection of the stabiliser polytope to this invariant space (Fig.~\ref{fig:proj_polytope}) still has exponentially many vertices, it turns out that formulating the optimisation problem in this symmetry-reduced way leads to a super-polynomially faster algorithm.

Equipped with the knowledge of the exact solution to Eq.~\eqref{eq:l1_robustness} for the commonly used magic states $\ket H^{\otimes n}$ and $\ket T^{\otimes n}$ and $n\leq 10$ qubits, we formulate a relaxation of the RoM problem for these states which yields an upper bound for the exact RoM. These approximations are in excellent agreement with the exact data for $n\leq 10$ and can be carried out for up to $26$ qubits. What is more, we can not only compute the RoM bounds for these approximations, but also find the corresponding affine decompositions $\rho^{\otimes n} = \sum_i x_i s_i$, which can directly be used in Monte Carlo simulations. Furthermore, we find a hierarchy of such RoM approximations by restricting to $k$-partite entangled stabiliser states which converges to the exact RoM. Interestingly, every level of the hierarchy can be computed in polynomial time.

Finally, both the exact and approximate results imply a runtime of $O(2^{0.737t})$ for simulating a circuit with $t$ T gates using the RoM algorithm. Moreover, our analysis suggests that this runtime is the optimal one that can be achieved using a RoM algorithm. Our work improves on the previously known runtime of $O(2^{0.753t})$ derived in Ref.~\cite{howard_application_2017}. Note that the RoM algorithm is able to simulate noisy circuits and mixed states. This is in contrast to simulation algorithms based on the so-called \emph{stabiliser rank} which can achieve a runtime of $O(2^{0.48t})$ for pure states \cite{bravyi_improved_2016,bravyi_trading_2016,bravyi_simulation_2018}. 

This paper is organised as follows. Section \ref{sec:rom} is devoted to a short discussion of the Robustness of Magic, giving an alternative definition to the one in the previous section and stating the properties of this resource monotone. Next, a series of techniques is presented which use the symmetries in the definition of the monotone to simplify the computation significantly. To this end, the symmetry group of the stabiliser polytope is characterised in Sec.~\ref{sec:symmetry_red} and certain classes of states are singled out in Sec.~\ref{sec:id_symmetries} which profit from a high degree of symmetry. For these states, we explicitly derive the symmetry-reduced problem by constructing a suitable basis for the invariant subspace in Sec.~\ref{sec:id_symmetries}, followed by enumerating equivalence classes of stabiliser states up to symmetry in Sec.~\ref{sec:representatives}. The numerical solutions for the constructed problems are presented and discussed in Section \ref{sec:computations}. Based on this, we prove a polytime relaxation of the RoM problem in Sec.~\ref{sec:rom_hierarchy}. Our results are summarised in Sec.~\ref{sec:conclusion}.

\begin{figure}
 \centering
 \subcaptionbox{$n=2$}{\includegraphics[width=0.4\textwidth]{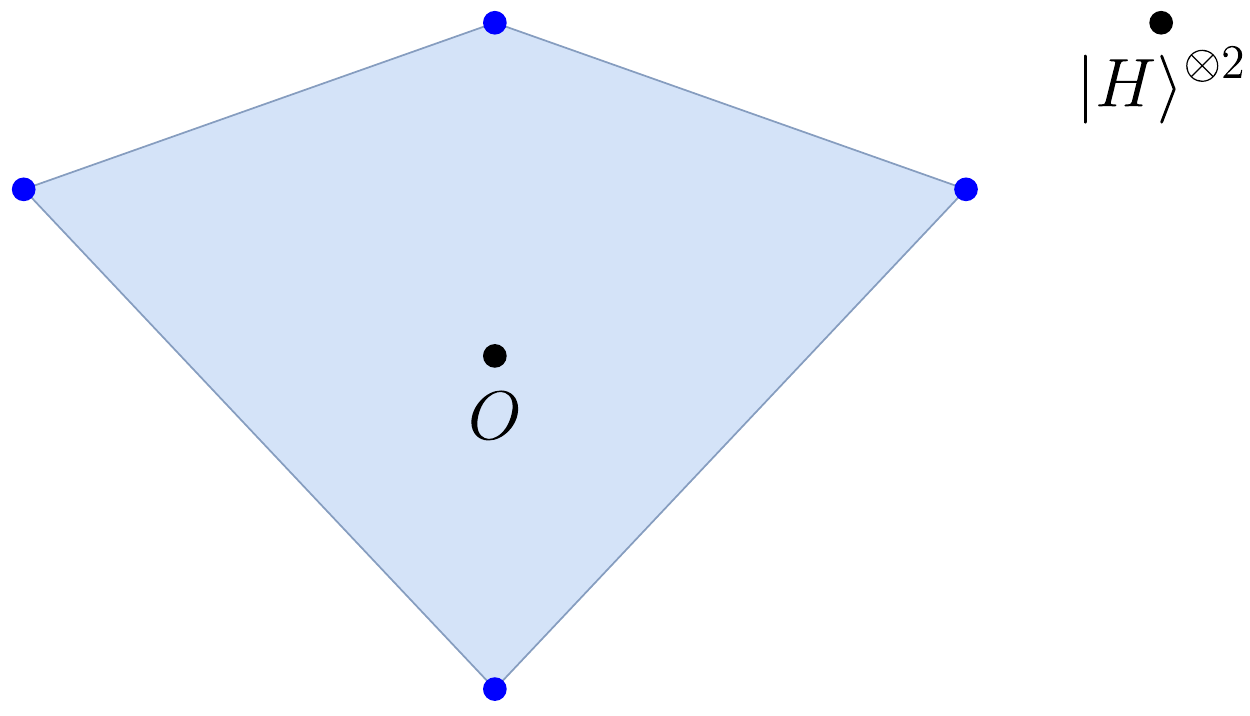}}
 \hspace{2em}
 \subcaptionbox{$n=3$}{\includegraphics[width=0.5\textwidth]{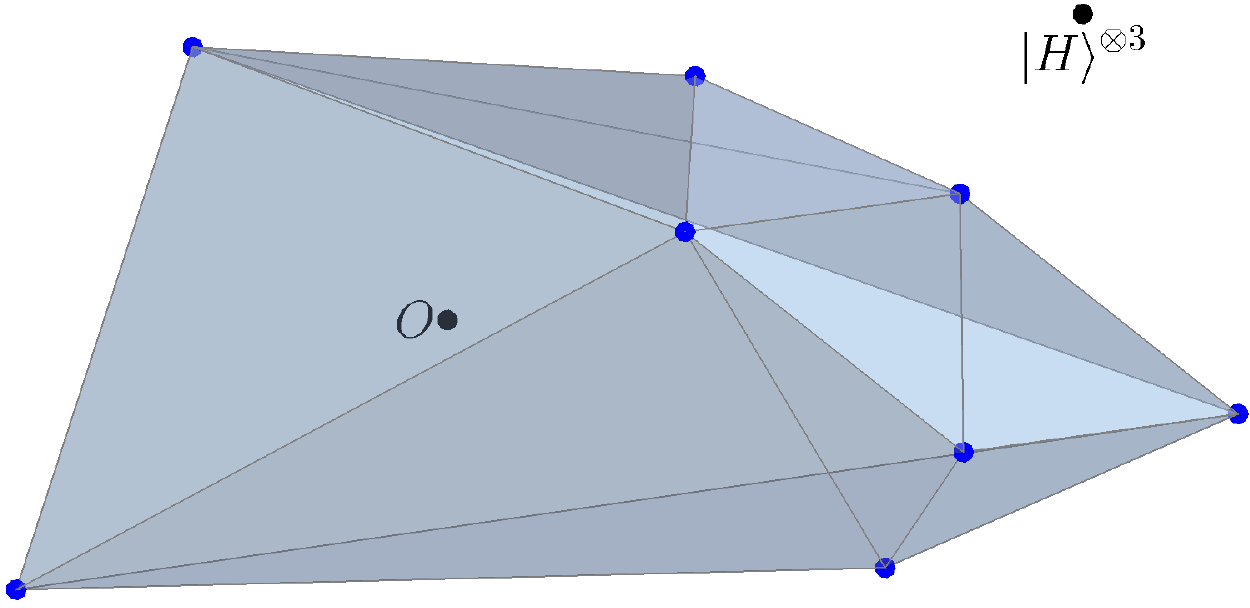}}
 \caption{Projected $n$-qubit stabiliser polytopes with respect to the symmetry group of the magic state $\ket H^{\otimes n}$ and $n=2,3$. We use a Bloch-like representation in the basis constructed in Sec.~\ref{sec:id_symmetries}. The origin $O$ corresponds to the maximally mixed state $\one/2^{n}$ and lies inside the polytope. The complexity of the polytopes is significantly reduced compared to the full 15-dimensional (respectively 63-dimensional) stabiliser polytopes. Visual inspection suggests that no joint symmetries of $\ket H^{\otimes n}$ and the projected polytope remain.
  }
 \label{fig:proj_polytope}
\end{figure}

\section{Robustness of Magic}
\label{sec:rom}

The resource theory of magic states can be developed in analogy to the more-established resource theory of entanglement and the \emph{robustness of entanglement} \cite{vidal_robustness_1999} studied in this context. 
There, the robustness of a state can be interpreted as a measure for the worst-case separable noise that renders the state separable.  However, its construction can be generalised to any resource theory as follows: Given a convex set $S$ of free resources, the robustness of $a$ relative to $b\in S$ is defined as
\begin{equation}
\label{eq:robustness}
 R(a||b) := \inf \left\{ s \geq 0 \; \bigg| \; \frac{1}{1+s}\left( a + s b \right) \in S \right\}.
\end{equation}
Depending on the choice of $b$, the robustness might be infinite. If it is finite, we can express $a$ as a pseudo-mixture 
\begin{equation}
 a = (1+s)b^+ - s b^-, \qquad\text{with }b^\pm\in S.
\end{equation}
Following \textcite{vidal_robustness_1999}, one can define the so-called \emph{total robustness} by minimising over the set of free resources:
\begin{equation}
 \label{eq:tot_robustness}
 R(a) := \inf_{b \in S}R(a||b).
\end{equation}

In the following, we choose $S=\SP_n$ to be the convex polytope spanned by the $n$-qubit stabiliser states.
More precisely, $\SP_n = \conv \stab(n)$, where $\stab(n)=\{s_1, \dots, s_N\}$ is the set of all $n$-qubit stabiliser states. 
Here, and in the following, by a ``quantum state'', we will always mean the density matrix representing it.
In the case of pure states $s_i = \ket{\psi_i}\bra{\psi_i}$, the associated vector $\ket{\psi_i}$ will be referred to as a \emph{state vector}.
The polytope $\SP_n$  is a subset of the real vector space of $(D\times D)$-dimensional Hermitian matrices $H_D$ where $D=2^n$ is the overall dimension of Hilbert space. More specifically, quantum states lie in the $(D^2-1)$-dimensional affine subspace given by $\tr\rho=1$. Within this affine hyperplane, $\SP_n$ is full-dimensional and we usually consider it as the the ambient space of $\SP_n$.

\textcite{howard_application_2017} work with an equivalent robustness measure: the \emph{robustness of magic} (RoM) introduced in Eq.~(\ref{eq:l1_robustness}).
A straightforward calculation (c.f.\ Appendix~\ref{sec:l1proof}) shows that the two measures are related by a simple affine transformation:
\begin{equation}
 \label{eq:l1_tot_robustness}
 \RM(\rho) = 1 + 2 R(\rho).
\end{equation}
The robustness of magic provides a proper resource monotone with the following  properties:
\begin{proposition}[Properties of Robustness of Magic \cite{howard_application_2017}]
 The \emph{robustness of magic} has the following properties:
 \begin{enumerate}
  \item \emph{Faithfulness}: $\RM(\rho)=1$ iff $\rho\in\SP_n$
  \item \emph{Monotonicity}: $\RM(\mathcal{X}(\rho)) \leq \RM(\rho)$ for all stabiliser operations $\mathcal{X}$ with equality if $\mathcal{X}$ is unitary.
  \item \emph{Convexity}: $\RM((1-t)\rho+t\sigma)\leq (1-t)\RM(\rho)+t\,\RM(\sigma)$ for $0\leq t \leq 1$.
  \item \emph{Submultiplicativity}: $\RM(\rho\otimes\sigma) \leq \RM(\rho)\,\RM(\sigma)$.
 \end{enumerate}
\end{proposition}

\section{Exploiting stabiliser symmetries}
\label{sec:symmetries}

\subsection{Definition of the RoM problem.}
\label{sec:rom_problem}

The Robustness of Magic is defined as the following optimisation problem.
\begin{problem}[Robustness of Magic]
\label{prob:rom}
Let  $\stab(n)=\{s_1,\dots,s_N\}$ be the set of stabiliser states.
Given a state $\rho$, solve the following problem:
\begin{align*}
  \textbf{min} \quad & \|x\|_1 \quad \text{over }x \in \R^{N} \\
  \textbf{s.\,t.}\quad & \rho = \sum_{i=1}^N x_i s_i. 
 \end{align*}
\end{problem}
Using standard techniques, this problem can be reformulated as a linear program (LP) with $D^2+2N$ constraints and $2N$ variables \cite{boyd_2009}. Although the time complexity of LPs is linear in the product of number of constraints and variables, these numbers themselves grow super-exponentially with the number of qubits $n$. Concretely, $N=2^{O({n^2})}$ and $D^2=4^n$. Moreover, the LP needs access to an oracle which provides the $N$ stabiliser states. The implementation of such an oracle would necessarily have super-exponential time complexity itself. However, even if an efficient oracle were provided, the storage of the states would quickly exceed the memory capacity of any computer. In
practice, this limits the evaluation of the problem to $n \leq 5$ on normal computers and renders it infeasible, even on supercomputers, for $n \geq 8$.\footnote{Already the storage of $\stab(7)$ would require around 77 TiB of memory.  For $n=8$, this number increases to around 76 PiB which exceeds the state-of-the-art by a factor of 7.}

A standard method in the analysis of optimisation problems is dualising the problem. Clearly, by Slater’s condition, strong duality holds and thus the dual problem is an equivalent definition for the Robustness of Magic. In Appendix \ref{sec:dual}, we state the dual problem and derive a lower bound from a feasible solution. However, this bound matches the one that was already found in Ref.~\cite{howard_application_2017}.

\subsection{Symmetry reduction}
\label{sec:symmetry_red}

The complexity of the RoM problem can be significantly reduced by exploiting the symmetries of the problem, a procedure that we will call \emph{symmetry reduction} and is well-known in convex optimisation theory, see \eg~ \cite{bachoc_2012}. 
Here, we will explain the basic ideas and refer the interested reader to App.~\ref{sec:symmetry_red_app} for a mathematical review.

By \emph{stabiliser symmetries} $\Aut(\SP_n)$, we mean the linear symmetry group of the stabiliser polytope.
This is the group of linear maps $H_D \to H_D$ that leave $\SP_n$ invariant. 
These maps necessarily have to preserve the set of vertices, \ie~the set of stabiliser states $\stab(n)$. 
Clearly, the group of $n$-qubit Clifford unitaries $\Cl_n$ induces such  symmetry transformations by conjugation. 
Another obvious symmetry of the set of stabilisers is the \emph{transposition}:
\begin{equation*}
	s_i = \ket{\psi_i}\bra{\psi_i} \mapsto s_i^T = \mathcal{C}\ket{\psi_i}\bra{\psi_i}\mathcal{C},
\end{equation*}
where $\mathcal{C}$ is the (anti-unitary) operation of complex conjugation in the computational basis.
The group of unitary and anti-unitary operations generated by Clifford unitaries and complex conjugation is known as the \emph{extended Clifford group} $\ECl_n$ \cite{appleby_symmetric_2005}.
Our first result states that any stabiliser symmetry is induced by the action of an element of the extended Clifford group on the Hilbert space. This is a corollary of the more general Thm.~\ref{thm:design_syms} on symmetries of 3-designs and is proven in App.~\ref{sec:designs}.

\begin{restatable}{corollary}{stabsym}
\label{cor:stabsymmetries}
 The group of stabiliser symmetries $\Aut(\SP_n)$ is given by the adjoint representation of the extended Clifford group $\ECl_n$.
\end{restatable}

We emphasise that this is a non-trivial result which is in general wrong for the case of odd-dimensional qudits where it is possible to construct explicit counter-examples. This turns out to be related to the fact that stabiliser states fail to form 3-designs in odd dimensions \cite{zhu_multiqubit_2017,kueng_qubit_2015,webb_clifford_2016}. 

Note that anti-unitary symmetries in $\ECl_n$ act in the adjoint representation as $\Ad(C)\circ T$, where $C\in\Cl_n$ and $T$ is the transposition map. Hence, there are only \emph{global} antiunitary symmetries. Every tensor product of local antiunitary symmetries would involve a partial transposition and such a map could not preserve the set of \emph{entangled} stabiliser states.

Let $\G_\rho < \ECl_n$ be a (not necessarily maximal) subgroup fixing $\rho$.
The projection onto the subspace of $\G_\rho$-fixed points $V^{\G_\rho}\subset H_D$, see App.~\ref{sec:symmetry_red_app}, is given by
\begin{equation}
 \Pi_{\G_\rho}(\sigma) = \frac{1}{|\G_\rho|} \sum_{U\in \G_\rho} U\sigma U^\dagger.
\end{equation}
Note that $\Pi_{\G_\rho}$ is trace-preserving, hence the image of quantum states will again lie in the affine subspace $\tr^{-1}(\{1\}) \cap V^{G_\rho}$.

Recall that we can express the robustness of $\rho$ as a minimisation over $t\geq 0$ and (mixed) stabiliser states $\sigma^{\pm}\in\SP_n$ such that
\begin{equation}
 \rho = (1+t)\sigma^+ - t \sigma^-.
\end{equation}
Since $\Pi_{\G_\rho}$ preserves $\SP_n$, every such decomposition yields a decomposition in terms of $\G_\rho$-invariant mixed stabiliser states:
\begin{equation}
 \rho = \Pi_{\G_\rho}(\rho) =  (1+t)\,\Pi_{\G_\rho}(\sigma^+) - t \, \Pi_{\G_\rho}(\sigma^-),
\end{equation}
In particular, if the decomposition was optimal in the first place, the projected decomposition is also optimal. 

This shows that there is always $\G_\rho$-invariant optimal solution for the problem. Hence, instead of optimising over the whole set of stabiliser states, we only have to optimise over $\G_\rho$-invariant mixed stabiliser states $\overline{\SP}_n:=\SP_n\cap V^{\G_\rho}$. By Lemma \ref{lem:poly_proj} in App.~\ref{sec:symmetry_red_app}, these are exactly given by $\overline{\SP}_n=\Pi_{\G_\rho}(\SP_n)$ and can thus be computed by evaluating the projections $\overline{\stab}(n):=\Pi_{\G_\rho}(\stab(n))$. Since $\Pi_{\G_\rho}(U s U^\dagger) = \Pi_{\G_\rho}(s)$ for all $U\in \G_\rho$ and $s\in\stab(n)$, it is sufficient to compute the projections on representatives of $\stab(n)/\G_\rho$. Finally, we remark that a majority of the projected states $\overline{\stab}(n)$ are not extremal points of the projected polytope $\overline{\SP}_n$. Given an extremal subset $\mathcal{V}_n=\{v_1,\dots,v_M\}\subset \overline{\stab}(n)$, the symmetry-reduced version of Prob.~\ref{prob:rom} is given by substituting $\stab(n)\mapsto\mathcal{V}_n$ and $N\mapsto M$.

\subsection{Identification of symmetries}
\label{sec:id_symmetries}

The first step towards the explicit symmetry-reduced problem is to identify the group $\G_\rho$ that fixes the state $\rho$ of interest. 
Motivated by magic state distillation and the submultiplicativity problem, we are especially interested in the case $\rho=\ket\psi\bra\psi^{\otimes n}$ with $\ket\psi$ being a $m$-qubit state. A large part of the analysis does not depend on the choice of $\ket\psi$, so we keep the discussion as general as possible and specialise later to $m=1$ and particular choices of $\ket\psi$. 
The symmetries of $\ket\psi^{\otimes n}$ can be classified as follows:

\begin{itemize}
 \item[] \textbf{Permutation symmetry} 
 Clearly, $\ket\psi^{\otimes n}$ is invariant under permutations of the $n$ tensor factors.
 Such permutations also preserve the stabiliser polytope.
 Thus, the symmetric group $S_n$ is contained in the symmetry group of the problem.
 \item[] \textbf{Local symmetries} 
 By local symmetries of $\ket\psi^{\otimes n}$ we mean products of $m$-qubit stabiliser symmetries of $\ket\psi$. 
 By Corollary \ref{cor:stabsymmetries}, this class contains only local Clifford operations.
 Let $(\Cl_m)_\psi$ be the stabiliser of $\ket\psi$ within the $m$-qubit Clifford group $\Cl_m$, then the local symmetry group is given by $(\Cl_m)_\psi^{\otimes n}$. 
 \item[] \textbf{Global symmetries} 
 We refer to all other symmetries as global. 
 The global symmetry group contains e.g.\ the transposition $\rho \mapsto \rho^T$.
\end{itemize}
The maximal symmetry group for $\rho=\ket\psi\bra\psi^{\otimes n}$ is given by the subgroup $\Cl_\rho$ that stabilises $\rho$ within $\ECl_n$. Here, we focus on the subgroup of $\Cl_\rho$ which is given by local symmetries and permutations:
\begin{equation}
 \G_{\rho} :=  (\Cl_m)^{\otimes n}_\psi \rtimes S_n.
\end{equation}
The following analysis suggests that for our choices of $\rho$, $\G_\rho$ actually coincides with $\Cl_\rho$, meaning that there are no further global symmetries. However, since the study of symmetries in $\ECl_n$ can be quite involved \cite{gross_schur-weyl_2017}, we can not exclude the possibility that we missed some of the symmetries. 

For the rest of this paper, we will consider the case $m=1$. Note that $\Cl_1$ acts by rotating about the symmetry axes of the stabiliser polytope. It is easy to see that states $\ket\psi$ with non-trivial stabilisers $(\Cl_1)_\psi$ fall into three classes: Stabiliser states (with trivial robustness), and magic states that lie on the Clifford orbit of $\ket H$ or $\ket T$.
Since the RoM is Clifford-invariant, we can pick the following states for concreteness:
\begin{equation}
 \ket{H}\bra{H} = \frac{1}{2}\left( \one + \frac{1}{\sqrt{2}} ( X + Y ) \right), \quad \ket{T}\bra{T} = \frac{1}{2}\left( \one + \frac{1}{\sqrt{3}} ( X + Y + Z ) \right).
\end{equation}
Figure \ref{fig:stab_polytope} shows the two states and their stabiliser symmetries. 
The respective unitary symmetries correspond to a two-fold rotation symmetry about the $\ket H$-axis and three-fold rotation symmetry about the $\ket T$-axis.
In terms of Clifford operations, these stabiliser groups are represented by
\begin{equation}
 (\Cl_1)_H = \langle SX \rangle, \qquad (\Cl_1)_T = \langle SH \rangle.
\end{equation}
Recall that these should be understood in the adjoint representation and thus the order of these groups is indeed $|(\Cl_1)_H|=2$ and $|(\Cl_1)_T|=3$. 

Furthermore, there are antiunitary stabiliser symmetries 
\begin{align}
 \mathcal A:\quad & X \mapsto X,  & \mathcal B:\quad & X \mapsto Y,     & \mathcal C:\quad & X \mapsto Z, \\
         & Y \mapsto Y,  &          & Y \mapsto X, &  & Y \mapsto Y, \\
         & Z \mapsto -Z, &          & Z \mapsto Z, &  & Z \mapsto X,
\end{align}
such that $\ket H$ is fixed by $\mathcal A$ and $\mathcal B$ and $\ket T$ is fixed by $\mathcal B$ and $\mathcal C$. Recall that these can only contribute global symmetries such as $\mathcal A^{\otimes n}$. However, the common $+1$ eigenspace of $\mathcal A^{\otimes n}$ and $\mathcal B^{\otimes n}$ coincides with that of $SX^{\otimes n}$ and thus adding these symmetries to the symmetry group will not further reduce the invariant subspace. A similar argument holds also for the antiunitary symmetries of $\ket T$. 

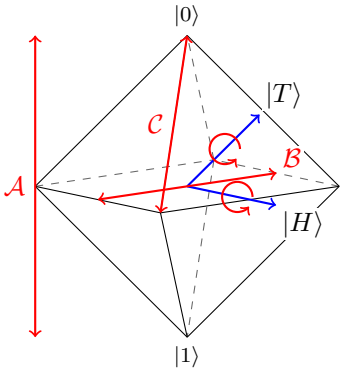
\begin{figure}[h]
\begin{minipage}[c]{0.4\linewidth}
 \begin{tikzpicture}[z=-5,scale=2]
  \coordinate (X1) at (1,0,0);
  \coordinate (X2) at (-1,0,0);
  \coordinate (Z1) at (0,1,0);
  \coordinate (Z2) at (0,-1,0);
  \coordinate (Y1) at (0,0,1);
  \coordinate (Y2) at (0,0,-1);
  \coordinate (H) at (0.71,0,0.71);
  \coordinate (T) at (0.5774,0.5774,0.5774);
  
  \draw [dashed,blue!50!white] (0,0,0) -- (0.5,0,0.5);

  \draw (X1) -- (Y1);
  \draw (X2) -- (Y1);

  \begin{scope}[dashed,opacity=0.6]
  \draw (Y2) -- (Z1);
  \draw (Y2) -- (Z2);
  \draw (X2) -- (Y2);
  \draw (X1) -- (Y2);
  \end{scope}

  \draw (X1) -- (Z1);
  \draw (X2) -- (Z1);
  \draw (Y1) -- (Z1);

  \draw (X2) -- (Z2);
  \draw (Y1) -- (Z2);
  \draw (X1) -- (Z2);
  
  \draw [->,blue,thick] (0,0,0) -- (H);
  \node [below right=1pt,inner sep=1pt,fill=white] at (H) {$\ket H$};
  \draw [->,blue,thick] (0,0,0) -- (T);
  \node [above right=1pt,inner sep=1pt,fill=white] at (T) {$\ket T$};
  \tdplotdrawarc[->,thick,red]{(0.4,0,0.4)}{0.1}{0}{320}{}{}
  \tdplotdrawarc[->,thick,red]{(0.3,0.3,0.3)}{0.1}{0}{320}{}{}
  
  \node [above] at (Z1) {\footnotesize $\ket{0}$};
  \node [below] at (Z2) {\footnotesize $\ket{1}$};

  \draw [<->,red,thick] (-0.5,0,0.5) --  (0.5,0,-0.5) node[above right=1pt,inner sep=1pt,fill=white] {$\mathcal B$};
  \draw [<->,red,thick] (-1,1,0) -- node[midway,left] {$\mathcal A$} (-1,-1,0);
  \draw [<->,red,thick] (Z1) -- node[midway,left] {$\mathcal C$} (Y1);
  
  \end{tikzpicture}
\end{minipage} 
\begin{minipage}[c]{0.5\linewidth}
\caption{Stabiliser symmetries of the magic states $\ket H$ and $\ket T$ and the octahedron of stabiliser states. $\ket H$ is fixed by the antiunitary reflections $\mathcal A$, $\mathcal B$ and unitary $\pi$ rotations around its axis. $\ket T$ is fixed by the antiunitary reflections $\mathcal B$, $\mathcal C$ and unitary $\pi/3$ rotations around its axis.}
\label{fig:stab_polytope}
\end{minipage} 
\end{figure}

Hence, the considered symmetry groups are as follows:
\begin{equation}
 \G_H := \langle SX \rangle^{\otimes n} \rtimes S_n , \qquad \G_T := \langle SH \rangle^{\otimes n} \rtimes S_n.
\end{equation}

Since the symmetric group $S_n$ is always a subgroup of the symmetry group, the  fixed point subspace $V^{\G_\rho}$ is always a subspace of the totally symmetric subspace $\Sym(H_D)$. Let us first consider a generic state $\rho$ with no further symmetries. Then, $V^{\G_\rho}$ coincides with $\Sym(H_D)$. Thus, the trace 1 subspace has dimension $\frac 1 6 (n+3)(n+2)(n+1)-1$ and is thus exponentially smaller than the full space. A basis for the symmetric subspace is given by a Fock-style ``occupation number basis'' constructed from the Pauli basis $\one, X, Y, Z$ as follows
\begin{multline}
 \label{eq:number_basis}
 N_{i,j,k} = \Sym\left( X^{\otimes i} \otimes Y^{\otimes j} \otimes Z^{\otimes k} \otimes \one^{\otimes(n-i-j-k)} \right), \\
 \text{for } i,j,k\in\{0,\dots,n\} \quad \text{such that} \quad i+j+k \leq n.
\end{multline}
Here, the symmetrisation operator $\Sym\equiv\Pi_{S_n}$ is given by averaging over all permutations of the tensor factors. The trace one subspace can be obtained as the span of all basis elements with the $N_{0,0,0}=\one$ component set to $1/D$. 

Due to linearity, the symmetrisation map is completely determined by its action on the Pauli basis. Given a Pauli operator $g$, there is a permutation $\pi\in S_n$ such that $\pi(g)=X^{\otimes i} \otimes Y^{\otimes j} \otimes Z^{\otimes k} \otimes \one^{\otimes(n-i-j-k)}$. The appearing exponents $i=\wt_X(g)$, $j=\wt_Y(g)$ and $k=\wt_Z(g)$ are exactly the \emph{weights} of $g$, \ie~the number of $X$, $Y$, $Z$ factors, respectively. By the invariance of $\Sym$ under permutations, we thus get $\Sym(g)=\Sym(\pi(g))=N_{i,j,k}$. We define \emph{weight indicator functions},
\begin{equation}
 \label{eq:wt_enumerator}
 A_{i,j,k}(g) := \begin{cases}
                    1   &   \text{if } \wt_X = i, \, \wt_Y = j, \, \wt_Z = k, \\
                    0   &   \text{else},
                   \end{cases}
\end{equation}
such that we can write the $S_n$-projection of a Pauli operator $g$ as 
\begin{equation}
 \Sym(g) = \sum_{i,j,k} A_{i,j,k}(g) N_{i,j,k}.
\end{equation}
By extending the functions $A_{i,j,k}$ linearly to $H_D$, we thus get exactly the coefficients of the projection in the number basis.

Let $S<\Pa_n$ be a stabiliser group stabilising a state $s$. The projection of this state is
\begin{equation}
\begin{split}
 \Sym(s) &= \frac{1}{2^n} \sum_{g\in S} \sgn(g) \Sym(g) \\  
 &= \frac{1}{2^n} \sum_{i,j,k} \left( \sum_{g\in S} \sgn(g) A_{i,j,k}(g) \right) N_{i,j,k} \\
&= \frac{1}{2^n} \sum_{i,j,k} A^\pm_{i,j,k}(S)\, N_{i,j,k}.
\end{split}
\end{equation}
The $A^\pm_{i,j,k}(S)$ are the coefficients of the \emph{complete signed quantum weight enumerators} of the stabiliser code $S$. 
Recall that for a classical code $C\subset \F_d^n$, the \emph{complete weight enumerator} is the degree-$n$ polynomial in $d$ variables given by
\begin{equation*}
	\sum_{c\in C} x_0^{\wt_0(c)} \dots x_{d-1}^{\wt_{d-1}(c)}
	=:
	\sum_{i_1, \dots, i_{d-1}} 
	A_{i_1, \dots, i_{d-1}}(C)\,
	x_0^{n-(i_1+\dots+i_{d-1})} x_1^{i_1} \dots x_{d-1}^{i_{d-1}},
\end{equation*}
where $\wt_i(c)$ gives the number of times $i\in\F_d$ appears in $c$ \cite{macwilliams_1977}.
The analogy should be clear.
\emph{Unsigned} weight enumerators for quantum codes have been studied since the early days of quantum coding theory \cite[Ch.~13]{nebe_self-dual_2006}.
Much less seems to be known about their signed counterparts, with Refs.~\cite{rall2017signed,rall_fractal_2017} being the only related references we are aware of. There it is shown that, as their classical analogues, signed quantum weight enumerators are NP-hard to compute.

Finally, we want to return to the cases $\ket\psi=\ket H$ and $\ket\psi=\ket T$ and discuss the invariant subspaces $V^{H,T}:=V^{\G_{H,T}}$ for these states. Let us rotate the Pauli basis such that the first basis vector corresponds to the Bloch representation of $\ket H$ and $\ket T$, respectively:
\begin{align}
  E_1^H &:= \frac{1}{\sqrt 2}\left( X + Y \right), & E_1^T &:= \frac{1}{\sqrt 3}\left( X + Y + Z \right), \\
  E_2^H &:= \frac{1}{\sqrt 2}\left( X - Y \right), & E_2^T &:= \frac{1}{\sqrt 6}\left( X - 2 Y + Z \right) \\
  E_3^H &:= Z, & E_3^T &:= \frac{1}{\sqrt 2}\left( X - Z \right).
\end{align}
Note that this choice of basis is such that the orthogonal decompositions of state space $H_2 = \langle \one \rangle \oplus \langle E_1^H \rangle \oplus \langle E_2^H \rangle \oplus\langle E_3^H \rangle = \langle \one \rangle \oplus  \langle E_1^T \rangle \oplus \langle E_2^T,E_3^T \rangle$ correspond to (real) irreps of the respective Clifford stabilisers $(\Cl_1)_{H,T}$, as can be seen from the matrix representation of the generators in the rotated basis:
\begin{equation}
 SX \simeq \begin{pmatrix}
            1 & 0 & 0 & 0 \\
            0 & 1 & 0 & 0 \\
            0 & 0 & -1 & 0 \\
            0 & 0 & 0 & -1
           \end{pmatrix},\qquad
 SH \simeq \frac 1 2
           \begin{pmatrix}
            2 & 0 & 0 & 0 \\
            0 & 2 & 0 & 0 \\
            0 & 0 & -1 & \sqrt 3 \\
            0 & 0 & -\sqrt 3 & -1
           \end{pmatrix}.
\end{equation}
In general, a basis for the trivial representation of $(\Cl_1)_{H,T}^{\otimes n}$ in the $n$-qubit state space $H_{2^n}$ is given by $\mathcal{B}^{H,T}:=\{\one, E_1^{H,T}\}^{\otimes n}$. To construct a basis for the full invariant subspace, we have to symmetrise $\mathcal{B}^{H,T}$ resulting in $\{N^{H,T}_{i,0,0}=:N^{H,T}_i\,|\,i=0,\dots,n\}$. Here, $N^{H,T}_{i,j,k}$  is the occupation number basis associated to the rotated basis $\{\one,E^{H,T}_1,E^{H,T}_2,E^{H,T}_3\}$ and is constructed analogously to before.

In general, the components of stabiliser states in the rotated bases can be written in terms of weight enumerators by computing the induced basis transformations on $\Sym(H_D)$ from $N_{i,j,k}$ to $N^{H,T}_{i,j,k}$. However, we are only interested in the projection onto $j=k=0$ which simplifies this computation. First, let us rewrite the $n$-qubit Pauli operators in the $H$-basis. Note that every operator with non-vanishing $Z$-weight is already in the orthocomplement of $V^H$.
\begin{equation}
\begin{split}
 X^{\otimes i} \otimes Y^{\otimes j} &= \left(\frac{1}{\sqrt{2}}\right)^{i+j} \left(E_1^H+E_2^H\right)^{\otimes i} \otimes \left(E_1^H-E_2^H\right)^{\otimes j}\\
 &= \left(\frac{1}{\sqrt{2}}\right)^{i+j} \left(E_1^H\right)^{\otimes (i+j)} + \text{orth. terms}
 \end{split}
\end{equation}
Here, we left out possible identity factors and all orthogonal terms on the RHS, \ie~those containing $E_2^H$. This result implies that we can write the projection of a stabiliser state $s$ as
\begin{equation}
 \label{eq:H_proj}
 \Pi_H(s) = \frac{1}{2^n} \sum_{i=0}^{n} \left(\sum_{j=0}^i A^\pm_{i-j,j,0}(S)\right) \frac{N^H_i}{2^{i/2}} =: \frac{1}{2^n} \sum_{i=0}^{n} B^\pm_i(S) \frac{ N^H_i}{2^{i/2}}.
\end{equation}
We call the numbers $B^\pm_i(S)$ the \emph{partial signed quantum weight enumerators} of $S$. The analysis works the same way for the $T$-projection:
\begin{equation}
\begin{split}
 X^{\otimes i} \otimes Y^{\otimes j}\otimes Z^{\otimes k} &= \left(\frac{E_1^T}{\sqrt{3}}+\frac{E_2^T}{\sqrt{6}}+\frac{E_3^T}{\sqrt{2}}\right)^{\otimes i} \otimes \left(\frac{E_1^T}{\sqrt{3}}-\sqrt{\frac{2}{3}}E_2^T\right)^{\otimes j}\otimes \\
 & \qquad\qquad \otimes \left(\frac{E_1^T}{\sqrt{3}}+\frac{E_2^T}{\sqrt{6}}-\frac{E_3^T}{\sqrt{2}}\right)^{\otimes k}\\
 &= \left(\frac{1}{\sqrt{3}}\right)^{i+j+k} \left(E_1^T\right)^{\otimes (i+j+k)} + \text{orth. terms}
 \end{split}
\end{equation}
In this case, the $T$-projection of a stabiliser state $s$ with stabiliser group $S$ involves \emph{total signed quantum weight enumerators} $C^\pm_i(S)$ as follows:
\begin{equation}
 \label{eq:T_proj}
 \Pi_T(s) = \frac{1}{2^n} \sum_{i=0}^{n} \left(\sum_{j=0}^i\sum_{k=0}^{i-j} A^\pm_{i-j-k,j,k}(S)\right) \frac{N^T_i}{3^{i/2}} =: \frac{1}{2^n} \sum_{i=0}^{n} C^\pm_i(S) \frac{N^T_i}{3^{i/2}}.
\end{equation}

Note that all projections $\Sym\equiv\Pi_{S_n}$, $\Pi_H$ and $\Pi_T$ can be computed from the complete signed weight enumerators of the stabiliser codes which themselves are functions of the weight distributions. For numerical purposes, it is convenient to absorb all appearing factors in the bases such that the coefficients of stabiliser states are given by the integer weight enumerators.

Finally, we want to give expressions for the states $\ket{H}^{\otimes n}$ and $\ket{T}^{\otimes n}$ in the respective bases:
\begin{align}
 \ket{H}\bra{H}^{\otimes n} &= \frac{1}{2^n}\left( \one + E_1^H \right)^{\otimes n} = \frac{1}{2^n}\sum_{i=0}^n \binom{n}{i} N^{H}_i, \\
 \ket{T}\bra{T}^{\otimes n} &= \frac{1}{2^n}\left( \one + E_1^T \right)^{\otimes n} = \frac{1}{2^n}\sum_{i=0}^n \binom{n}{i} N^{T}_i.
\end{align}

In general, we are not aware of any method which can predict whether the projection of a stabiliser state will be extremal within the projected polytope. However, the following lemma gives a necessary condition on the extremality of products $s \otimes s'$ of stabiliser states which will be useful later.

\begin{lemma}[Projection of product states]
\label{lem:proj_prod_states}
The following is true for $\Pi=\Sym,\Pi_H,\Pi_T$: If the projection $\Pi(s)$ of an arbitrary stabiliser state $s$ is non-extremal, so is $\Pi(s \otimes s')$ for any other stabiliser state $s'$.
\end{lemma}
\begin{proof}
We prove the statement by showing it on the level of the complete signed weight enumerators $A^\pm_{i,j,k}$. This proves the claim directly for $\Pi=\Sym$ and the other cases follow since the partial and total signed weight enumerators are linear functions of the complete ones. 

Note that the Pauli $X$, $Y$, $Z$ weights are additive under tensor products, \eg~$\wt_X(g\otimes g') = \wt_X(g) + \wt_X(g')$. This implies that we can write the indicator function as $A_{i,j,k}(g\otimes g')=A_{i',j',k'}(g)A_{i-i',j-j',k-k'}(g')$ for $i'$, $j'$, $k'$ being the weights of $g$. However, since $A_{i',j',k'}(g)$ is zero if $i'$, $j'$, $k'$ are \emph{not} the weights of $g$, we can instead sum over all possible decompositions on the right hand side.
Hence, for any two stabiliser codes $S,S'$ we get
\begin{equation}
 \begin{split}
 \label{eq:wt_prod_states}
  A^\pm_{i,j,k}(S\times S') &= \sum_{g\in S,g'\in S'} \sgn(g\otimes g') A_{i,j,k}(g\otimes g') \\ 
   &= \sum_{i'=0}^{i}\sum_{j'=0}^{j}\sum_{k'=0}^{k}  \sum_{g\in S}\sum_{g'\in S'} \sgn(g)\sgn(g') A_{i',j',k'}(g) A_{i-i',j-j',k-k'}(g') \\
   &= \sum_{i'=0}^{i}\sum_{j'=0}^{j}\sum_{k'=0}^{k} A^\pm_{i',j',k'}(S) A^\pm_{i-i',j-j',k-k'}(S').
 \end{split}
\end{equation}
Suppose $S$ is the stabiliser of a state $s$ and $\Sym(s)$ can be written as convex combination,
 \begin{equation}
  \Sym(s) = \sum_{l=1}^M \lambda_l \Sym(s_l) \quad \Leftrightarrow \quad A^\pm_{i,j,k}(S) = \sum_{l=1}^M \lambda_l A^\pm_{i,j,k}(S_l),
 \end{equation}
 with stabiliser states $s_l$, stabilised by the groups $S_l$. Let $s'$ be stabilised by $S'$, then we find by Eq.~\eqref{eq:wt_prod_states},
 \begin{equation}
  A^\pm_{i,j,k}(S \times S') = \sum_{i',j',k'} \sum_{l=1}^M \lambda_l A^\pm_{i',j',k'}(S_l) A^\pm_{i-i',j-j',k-k'}(S')
   = \sum_{l=1}^M \lambda_l A^\pm_{i,j,k}(S_l \times S'),
 \end{equation}
 and hence the projection of the product state $s\otimes s'$ is non-extremal.
\end{proof}

Note that Eq.~\eqref{eq:wt_prod_states} allows us to compute the projection of products $\Pi(s\otimes s')$ from $\Pi(s)$ and $\Pi(s')$ via the signed quantum weight enumerators using $\poly(n)$ operations. This is an important improvement over computing $\Pi(s)$ for a general (fully entangled) stabiliser state $s$ which requires $O(2^n)$ operations.

\subsection{Representatives of inequivalent stabiliser states}
\label{sec:representatives}

Computing the projected polytope involves the computation of the signed quantum weight enumerators for all stabiliser states. However, from the previous discussions we know that we can restrict the computations to the orbits $\stab(n)/\G_\rho$ with respect to the symmetry group $\G_\rho.$ In this section we will construct representatives for these orbits.

Our approach is based on a subset of the set of stabiliser states, the so-called \emph{graph states} $\graph(n)$. For every simple, \ie~self-loop free, graph $G$ of $n$ vertices, there is a state vector $\ket G$ that is stabilised by operators of the form
\begin{equation}
\label{:eq:graph_generators}
 K_j = X_j \prod_{k=1}^n Z_k^{\theta_{jk}}, \qquad (j=1,\dots,n),
\end{equation}
where $X_j,Z_j$ are the Pauli operators on the $j$-th qubit and $\theta$ is the adjacency matrix of the graph $G$. Graph states play a fundamental role in the studies of stabiliser states since \textcite{schlingemann_stabilizer_2001} proved that every stabiliser state is equivalent to a graph state under the action of the local Clifford group $L\Cl_n = \Cl_1^{\otimes n}$:
\begin{equation}
 \stab(n) = L\Cl_n \cdot \graph(n).
\end{equation}
This result can be used to label every stabiliser state vector $\ket{C,G}$ by a local Clifford unitary $C\in L\Cl_n$ and a graph state $\ket{G}\in\graph(n)$ such that $\ket{C,G}= C\ket{G}$. However, $L\Cl_n$-equivalent graph states generate the same $L\Cl_n$-orbit and are equally well suited to represent a stabiliser state. \textcite{nest_graphical_2004,hein_multiparty_2004} discovered that that two graph states are $L\Cl_n$-equivalent if and only if the underlying graphs are related  by a graph theoretic transformation called \emph{local complementation} (LC). Thus, it is sufficient to consider graphs up to local complementation. 

Furthermore, the symmetry group $\G_\rho$ induces additional equivalence relations on the graph state representation. Let us again begin the discussion with the case of a generic state with $S_n$-symmetry. This already allows us to restrict the representation to non-isomorphic graphs, \ie~graphs up to permutation of their vertices, since for any graph state $\ket{G}$ and a permuted version $\ket{\pi G} \equiv \pi\ket{G}$ the $L\Cl_n$-orbits are isomorphic: $\pi C\ket{G} = C_\pi \ket{\pi G}$ with the permuted local Clifford unitary $C_\pi = \pi C \pi^\dagger \in L\Cl_n$. Moreover, it is straightforward to show that the composition of graph isomorphism and local complementation is symmetric and thus a equivalence relation $\sim_{LC,S_n}$ on graphs whose equivalence classes are isomorphic to  $\graph(n)/\sim_{L\Cl_n,S_n}$. These equivalence classes have been studied in the context of graph codes and entanglement in graph states \cite{danielsen_classification_2006,hein_entanglement_2006} and were enumerated by \textcite{danielsen_website}. However, different local Clifford unitaries can still result in equivalent states. To see this, pick some symmetry $\pi\in \Aut(G)$ of the graph, \ie~$\pi G = G$, then the actions of $C$ and $C_\pi$ yield isomorphic states. Hence, it is enough to act with $L\Cl_n/\Aut(G)$ on the graph state $\ket G$.

For the computation of the $L\Cl_n$-orbits it is enough to consider $L\Cl_n/\Pa_n$, since Pauli operators will only change the possible $2^n$ signs of the final generators which are better added by hand. It is well known that the quotient $\Cl_n/\Pa_n$ is isomorphic to the binary symplectic group $\Sp(2n,\Z_2)$ which is the foundation of the phase space formalism. We make use of this formalism to compute the $L\Cl_n$-orbits of graph states $G$ by evaluating the orbits of the local symplectic group $\Sp(2,\Z_2)^{\times n}$ up to the stabiliser of $G$ and $\Aut(G)$.

The additional symmetries in the case of the $\ket H$ and $\ket T$ state can be taken into account by restricting the allowed symplectic transformations using the symplectic maps $\hat S$ and $\hat{S}\hat{H}$ induced by the generators $SX$ and $SH$, respectively. The corresponding cosets are given by the representatives $\Sp(2,\Z_2)/\langle \hat S \rangle \simeq \{\one, \hat H, \hat H \hat S\}$ and $\Sp(2,\Z_2)/\langle \hat{S}\hat{H} \rangle \simeq \{ \one, \hat S \}$, respectively.

However, the described generation procedure will quickly become computationally expensive.  Moreover, most of the projected stabiliser states are non-extremal points for the projected polytope and thus redundant. Unfortunately, there is no simple way of deciding whether a state will be extremal after projection or not. However, Lemma \ref{lem:proj_prod_states} states at least a criterion for product states which allows us to restrict to projecting only \emph{fully entangled stabiliser states}. To this end, we only have to iterate over \emph{connected} graph representatives with respect to $\sim_{LC,S_n}$ and compute the projections of product states directly from lower-dimensional vertices using the appropriate version of Eq.~\eqref{eq:wt_prod_states}.

\section{Computing the robustness of magic}
\label{sec:computations}

Using the enumeration procedure of the last section, we generated the set of $H$- and $T$-projections of \emph{fully entangled} stabiliser states $\overline{\stab}^{H/T}_c(n)=\Pi_{H/T}(\stab_c(n))$ and the set of projected product states from lower-dimensional vertices. In an additional step, we removed non-extremal points from the set of projected states, resulting in vertex sets $\mathcal{V}^{H/T}_n$ of the projected stabiliser polytopes for $n\leq 9$ and $n\leq 10$, respectively. As described in the last section, we are labelling the vertices by certain stabiliser representatives. To this end, we use a notation in terms of ``decorated graph states'' compatible with Refs.~\cite{schlingemann_stabilizer_2001,elliott_graphical_2008}: A graph is decorated by symbols which indicate the action of local Clifford operations on the respective graph state. Nodes with signs indicate a sign change of the respective stabiliser generator, or alternatively, the action of $Z$ on the respective qubit prior the any other gates. A hollow node in the graph denotes a Hadamard gate acting on the respective qubit and self-loops correspond to the action of phase gates (prior to possible Hadamard gates). Figure \ref{fig:Hvertices} shows the vertex sets $\mathcal{V}^H_n$ for $n=1,2,3$. Since the dimension of the polytope is exactly $n$, it can be easily visualised for $n\leq 3$, see also Fig.~\ref{fig:proj_polytope} in Sec.~\ref{sec:rom}. 

The database of vertices and the program code can be found on the arXiv \cite{heinrich_robustness_2018}. For a discussion of the algorithmic details see App.~\ref{sec:numerics}.

\begin{figure}
\newlength{\myheight}
\setlength{\myheight}{1.8cm}

\centering

\subcaptionbox{$n=1$}{%
  \includegraphics[height=0.5\myheight]{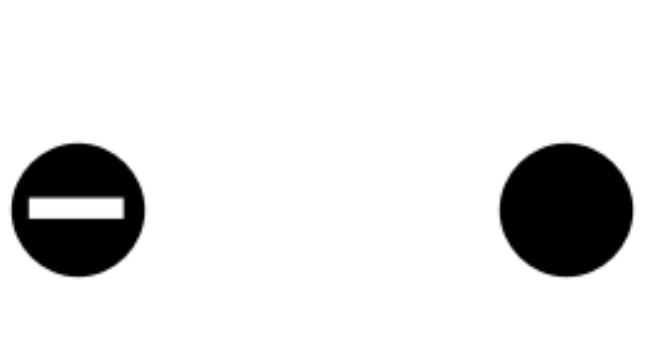}%
}\hspace{4em}
\subcaptionbox{$n=2$}{%
  \includegraphics[height=\myheight]{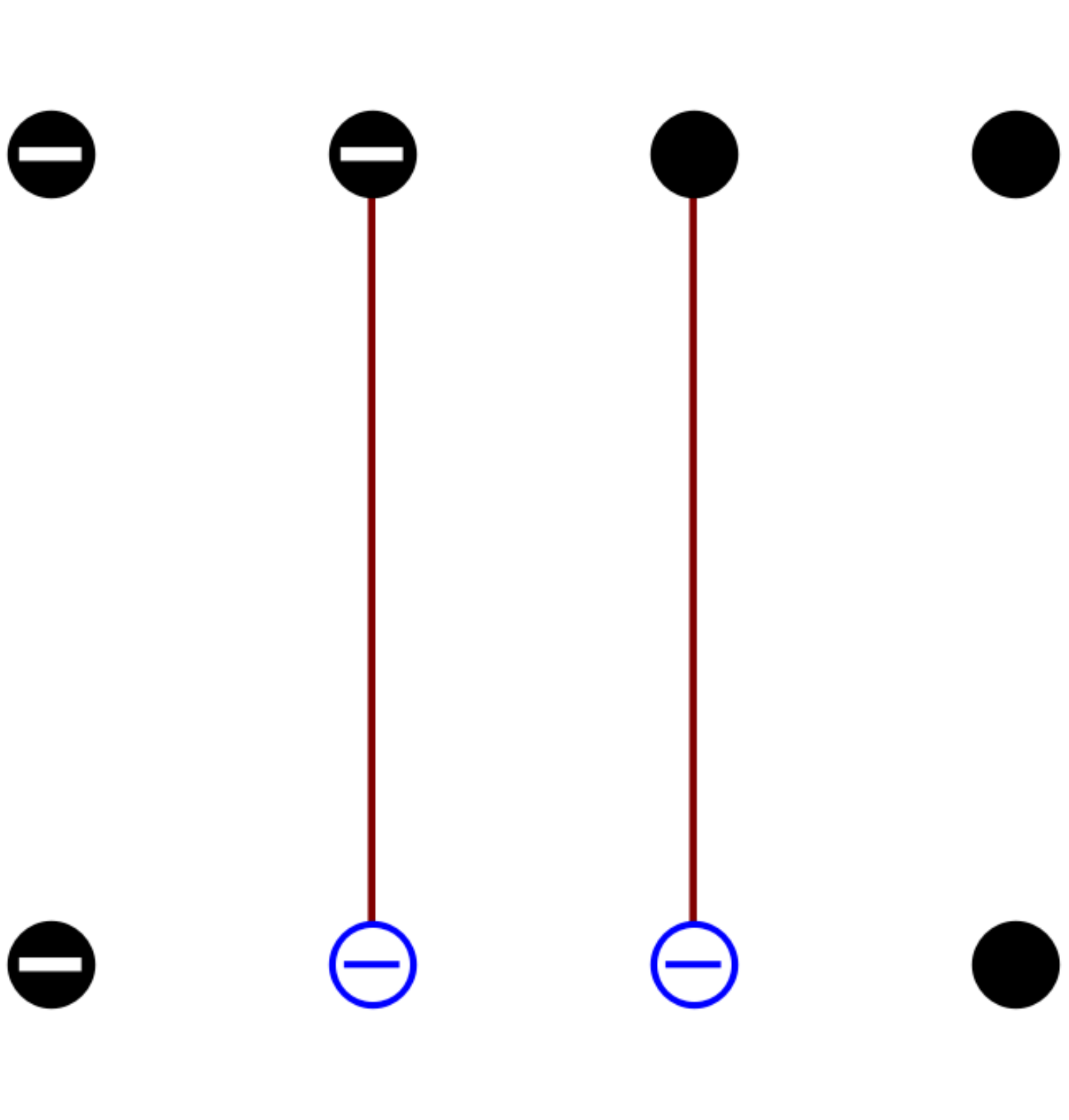}%
}\\[2em]
\subcaptionbox{$n=3$}{%
  \includegraphics[height=2\myheight]{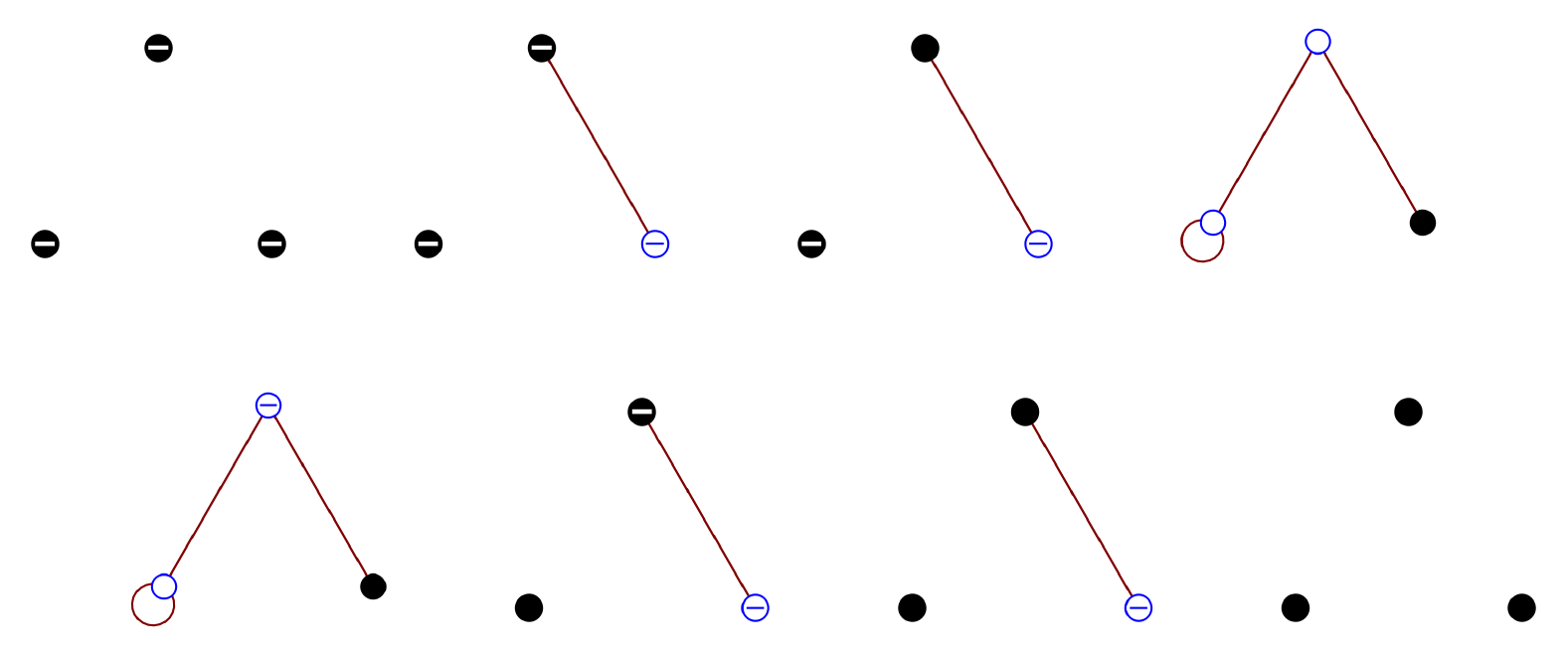}%
}
\caption{Vertices of the projected stabiliser polytope for the $\ket{H}$ symmetry group $\G_H$. These states are represented as decorated graph states compatible with Refs.~\cite{schlingemann_stabilizer_2001,elliott_graphical_2008}, c.f.~the description in the text. The convex hull of these vertices is shown in Fig.~\ref{fig:proj_polytope}.}
\label{fig:Hvertices}
\end{figure}

\begin{table}
\centering
 \begin{tabular}{l|l|l|l|l|l}
   $n$ & $\left|\stab(n)\right|$ & $\big|\overline{\stab}^{H}_c(n)\big|$ + prod. & $\left|\mathcal{V}^H_n\right|$ & $\big|\overline{\stab}^{T}_c(n)\big|$ + prod. & $\left|\mathcal{V}^T_n\right|$  \\ \hline
   1 & 6 & 3+0 & 2 & 2+0 & 2 \\
   2 & 60 & 5+3 & 4 & 4+3 & 4 \\
   3 & 1080 & 11+8 & 8 & 4+8 & 6 \\
   4 & 36720 & 48+18 & 13 & 18+14 & 12 \\
   5 & 2423520 & 252+38 & 32 & 61+26 & 22 \\
   6 & 315057600 & 1881+86 & 60 & 256+57 & 42 \\
   7 & 81284860800 & 20378+208 & 144 & 2151+116 & 66 \\
   8 & 41780418451200 & 331794+510 & 304 & 21475+226 & 131 \\
   9 & 42866709330931200 & 8410183+1270 & 804 & 329712+462 & 238 \\
   10 & --  &   --  &   --  &   5964000+991 &   371
 \end{tabular}
\caption{Number of stabiliser states $|\stab(n)|$ in comparison with the number of projections of fully entangled stabiliser states $|\overline{\stab}_c(n)|$, projected product states and vertices $|\mathcal{V}_n|$ of the projected stabiliser polytope as a function of the number of qubits $n$.}
\label{tab:vertices}
\end{table}

Table \ref{tab:vertices} shows the number of vertices of the projected polytopes in comparison with the original number of stabiliser states. We see that the number of states $N$ that have to be used in the $\ell_1$-minimisation is reduced drastically from $2^{O(n^2)}$ to a scaling which is approximately $2^n$. Additionally, the dimension $d$ of the ambient space is reduced exponentially from $4^n-1$ to exactly $n$. As discussed in Sec.~\ref{sec:symmetry_red}, the required $\ell_1$-minimisation for RoM is computed via a linear program with $2N+d$ constraints and $2N$ variables and has a runtime that is linear in its size $(2N+d)(2N)=4N^2+2Nd$. The runtime is thus reduced as
\begin{equation}
 2^{O(n^2)} \longrightarrow 2^{O(n)},
\end{equation}
leading to a super-polynomial speed-up in the $\ell_1$-minimisation. Although both time and space complexity of the $\ell_1$-minimisation are exponential in $n$, it is in principle  feasible for moderate $n$. Here, the limiting factor is the implementation of the oracle providing the projected states with runtime which is still super-exponential in $n$.

\subsection{Robustness of the $\ket H^{\otimes n}$ and $\ket T^{\otimes n}$ states}
\label{sec:rom_H}

Figure \ref{fig:rom_plot} shows the Robustness of Magic of $\ket{H}^{\otimes n}$ for $n=1,\dots,9$, computed from the vertices $\mathcal V_n^H$ of the projected stabiliser polytope. Note that the data for $n\leq 5$ is in perfect agreement with the so-far computed values in Ref.~\cite{howard_application_2017}. We are particularly interested in the submultiplicative behaviour of $\RM$. Here, the new data for $n>5$ turns out to be helpful: We can observe that the data points quickly approach an apparent exponential scaling with $n$. More precisely, submultiplicativity is clearly observable for $1 \leq n \leq 4$, but the scaling becomes effectively multiplicative for larger $n$. We quantified this using an exponential fit of the data range $3 \leq n \leq 9$ (shown in blue in Fig.~\ref{fig:rom_plot}) resulting in $(1.059\pm 0.015)\times (1.283\pm0.002)^n$. From previous works it is known that the \emph{regularised} robustness $\RM_\mathrm{reg}(\ket H)$ is bounded from below by 1.207. Our work, however, indicates that it converges from above to a constant which is given by the fit as $(1.283\pm0.002)$.

\begin{figure}[h]
\begin{minipage}{0.68\textwidth}
 \input{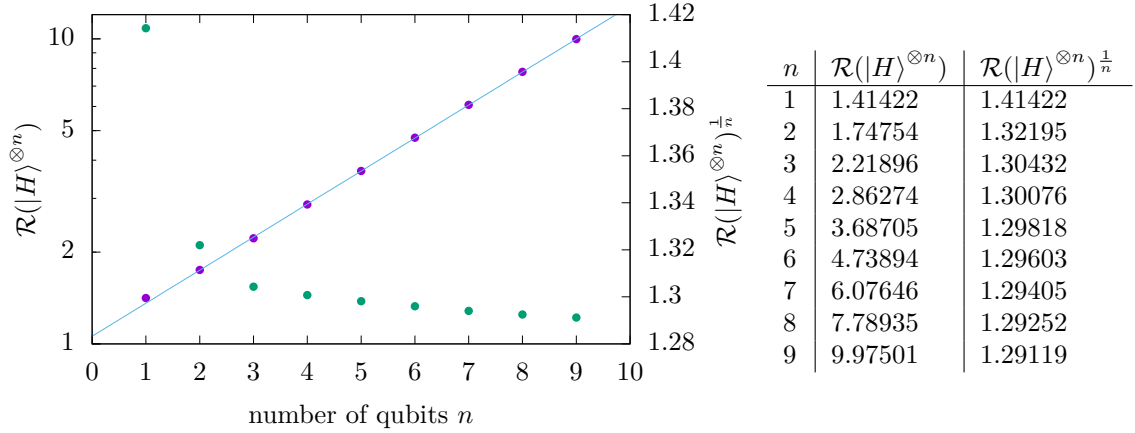}
\end{minipage}
\begin{minipage}{0.2\textwidth}
 \begin{tabular}{l|l|l}
   $n$ & $\mathcal{R}(\ket H^{\otimes n})$ & $\mathcal{R}(\ket H^{\otimes n})^{\frac 1 n}$ \\ \hline
   1 & 1.41422 & 1.41422 \\
   2 & 1.74754 & 1.32195 \\
   3 & 2.21896 & 1.30432 \\
   4 & 2.86274 & 1.30076 \\ 
   5 & 3.68705 & 1.29818 \\
   6 & 4.73894 & 1.29603 \\
   7 & 6.07646 & 1.29405 \\
   8 & 7.78935 & 1.29252 \\
   9 & 9.97501 & 1.29119
 \end{tabular}
\end{minipage}
\caption{Robustness (blue) and regularised robustness (green) of the magic state $\ket{H}^{\otimes n}$ as a function of the number of qubits $n$. The blue line is the exponential fit $(1.059\pm 0.015)\times (1.283\pm0.002)^n$ of the data.}
\label{fig:rom_plot}
\end{figure}

The previously known time complexity for simulating a circuit with $t$ T gates using the RoM algorithm is $O(2^{0.753t})$ \cite{howard_application_2017}. Our findings improve this to $O(\RM(\ket H^{\otimes 9})^\frac{2t}{9})=O(1.667^t)=O(2^{0.737t})$. Moreover, since we already explored an effectively multiplicative regime of the RoM, solving the problem for higher $n > 9$ will not much reduce the runtime. From our estimate for the asymptotic regularised robustness, we can estimate the best possible scaling to be $2^{0.719t}$.

Furthermore, we applied the same procedure to compute the robustness of the magic state $\ket T^{\otimes n}$. Since the $T$-symmetry group is larger than in the previous case, we were able to compute $\RM(\ket T^{\otimes n})$ for up to 10 qubits, see Fig.~\ref{fig:rom_plotT}. Qualitatively, the results agree very well with those of the last section. Quantitatively, the robustness of the $T$ state is considerably higher than the one of the $H$ state. Using again an exponential fit, we find the scaling $(1.169\pm0.011)\times(1.3865\pm0.0014)^n$ which predicts a regularised robustness of $(1.3865\pm0.0014)^n$. By the RoM construction, the 10-qubit solution gives rise to a simulation algorithm with runtime $O(1.984^m)=O(2^{0.988 m})$ where $m$ is the total number of $\ket{T}$ magic states used, or equivalently, the number of $\pi/12$ $Z$-rotation gates.

\begin{figure}
\begin{minipage}{0.68\textwidth}
 \input{gfx/i_romplotT}
\end{minipage}
\begin{minipage}{0.2\textwidth}
 \begin{tabular}{l|l|l}
   $n$ & $\mathcal{R}(\ket T^{\otimes n})$ & $\mathcal{R}(\ket T^{\otimes n})^{\frac 1 n}$ \\ \hline
    1 & 1.73206 & 1.73206 \\
    2 &	2.23206 & 1.49401 \\
    3 &	3.09808 & 1.45780 \\
    4 &	4.33100 & 1.44260 \\
    5 & 6.04494 & 1.43311 \\
    6 &	8.35898 & 1.42460 \\
    7 &	11.5114 & 1.41772 \\
    8 &	15.8436 & 1.41248 \\
    9 & 22.1823 & 1.41110 \\
    10 & 30.7056 & 1.40839
 \end{tabular}
\end{minipage}
\caption{Robustness (blue) and regularised robustness (green) of the magic state $\ket{T}^{\otimes n}$ as a function of the number of qubits $n$. The blue line is the exponential fit $(1.169\pm0.011)\times(1.386\pm0.0014)^n$ of the data.}
\label{fig:rom_plotT}
\end{figure}

\subsection{Analysis of the optimal solutions}
\label{sec:rom_sol}

Additionally, we studied the optimal solutions of the $\ell_1$-minimisation for the previously discussed cases of $\ket H^{\otimes n}$ and $\ket T^{\otimes n}$. For this purpose, it is instructive to use the original formulation of the robustness of a state $\rho$ in terms of an optimal affine combination of two (mixed) stabiliser states $\sigma^\pm\in\overline\SP^{H,T}_n$, cp.~Eq.~\eqref{eq:tot_robustness}:
\begin{equation}
 \rho = \frac{1}{2} \left[ (\RM(\rho)+1) \sigma^+ - (\RM(\rho)-1) \sigma^- \right].
\end{equation}
The states $\sigma^\pm$ can be obtained from the optimal solution of the $\ell_1$-minimisation $\rho=\sum_i x^*_i v_i$ as follows:
\begin{equation}
 \sigma^+ = \frac{2}{\RM(\rho)+1} \sum_{i:\, x^*_i > 0} x^*_i v_i, \qquad \sigma^- = -\frac{2}{\RM(\rho)-1} \sum_{i:\, x^*_i < 0}  x^*_i v_i.
\end{equation}
Recall from the discussion in Sec.~\ref{sec:symmetry_red} that replacing every vertex $v_i$ in the optimal solution by a stabiliser representative in its preimage $\Pi_{H,T}^{-1}(v_i)$ yields an optimal solution for the original problem. Hence, we simply identify the vertices of the projected polytope by their stabiliser representatives constructed in Sec.~\ref{sec:id_symmetries}. Surprisingly, these states seem to have a rather simple structure, especially the positive contributions $\sigma^+$. We will discuss the solutions in the following for the $H$ and $T$ case separately.

\paragraph{Optimal solutions for the $\ket H^{\otimes n}$ state} 

The positive contributions $\sigma^+$ to the $\ket H^{\otimes n}$ state for $n=1,2,3$ are simply given by the graph state $\ket +^{\otimes n}$. Figure \ref{fig:plus_states} shows the remaining states for $n=4,\dots,8$. Note that these states have to lie on a facet of the polytope to minimise the robustness. But instead of the generic $n$ contributions, they can be written using only $\lfloor\log_2 n\rfloor$ terms. The vertices themselves are products of $\ket +$ and the Bell state $\ket{\Psi^+}$.

\begin{figure}

\setlength{\myheight}{1.8cm}

\centering

\subcaptionbox{$n=4$}{%
  \includegraphics[height=\myheight]{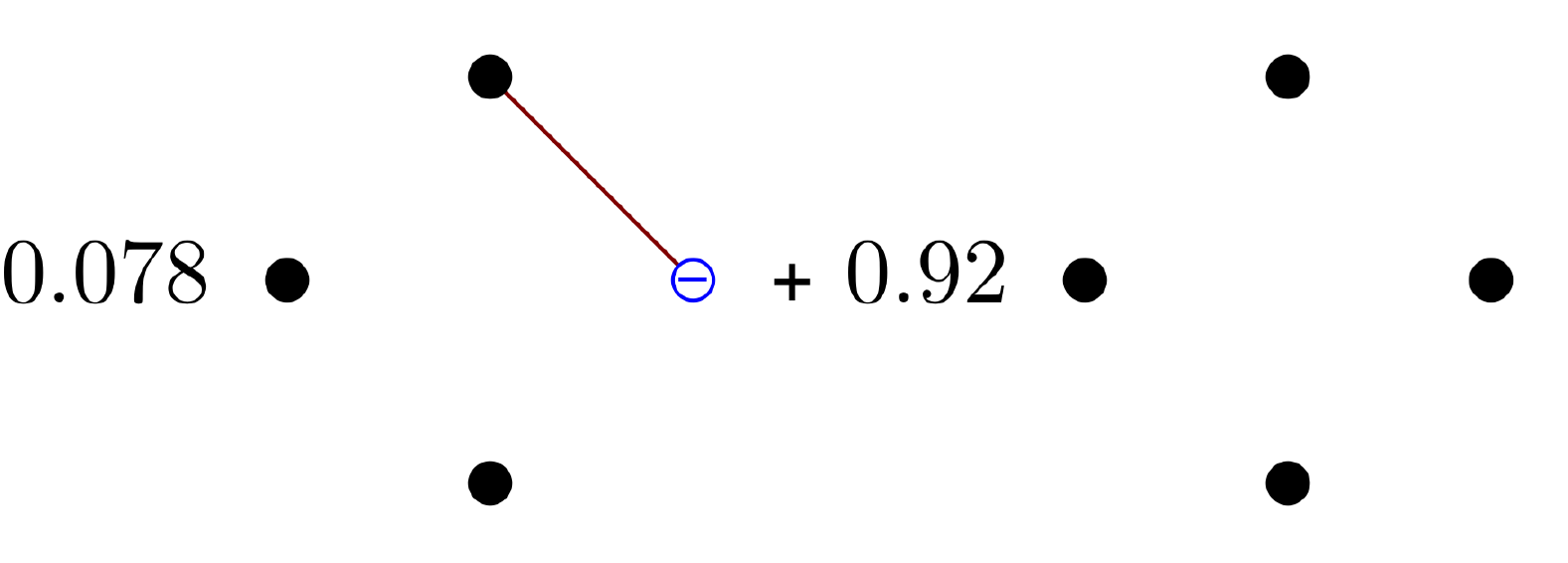}%
}\hspace{4em}
\subcaptionbox{$n=5$}{%
  \includegraphics[height=\myheight]{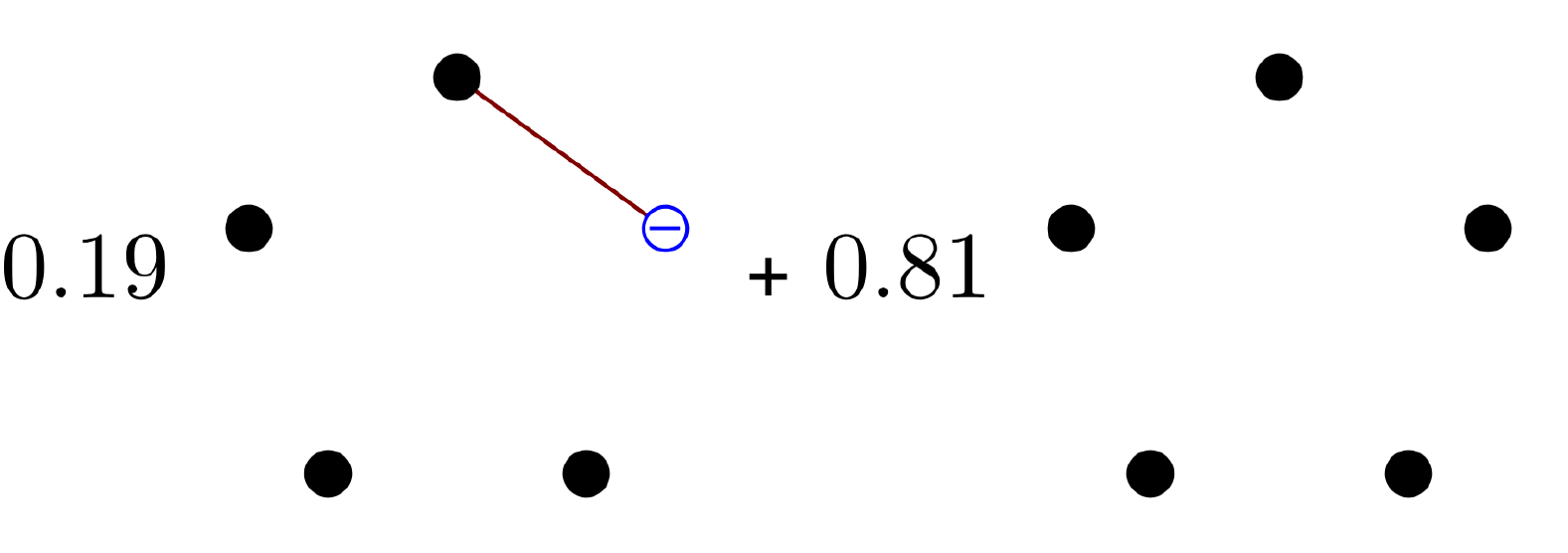}%
}\\[2em]
\subcaptionbox{$n=6$}{%
  \includegraphics[height=\myheight]{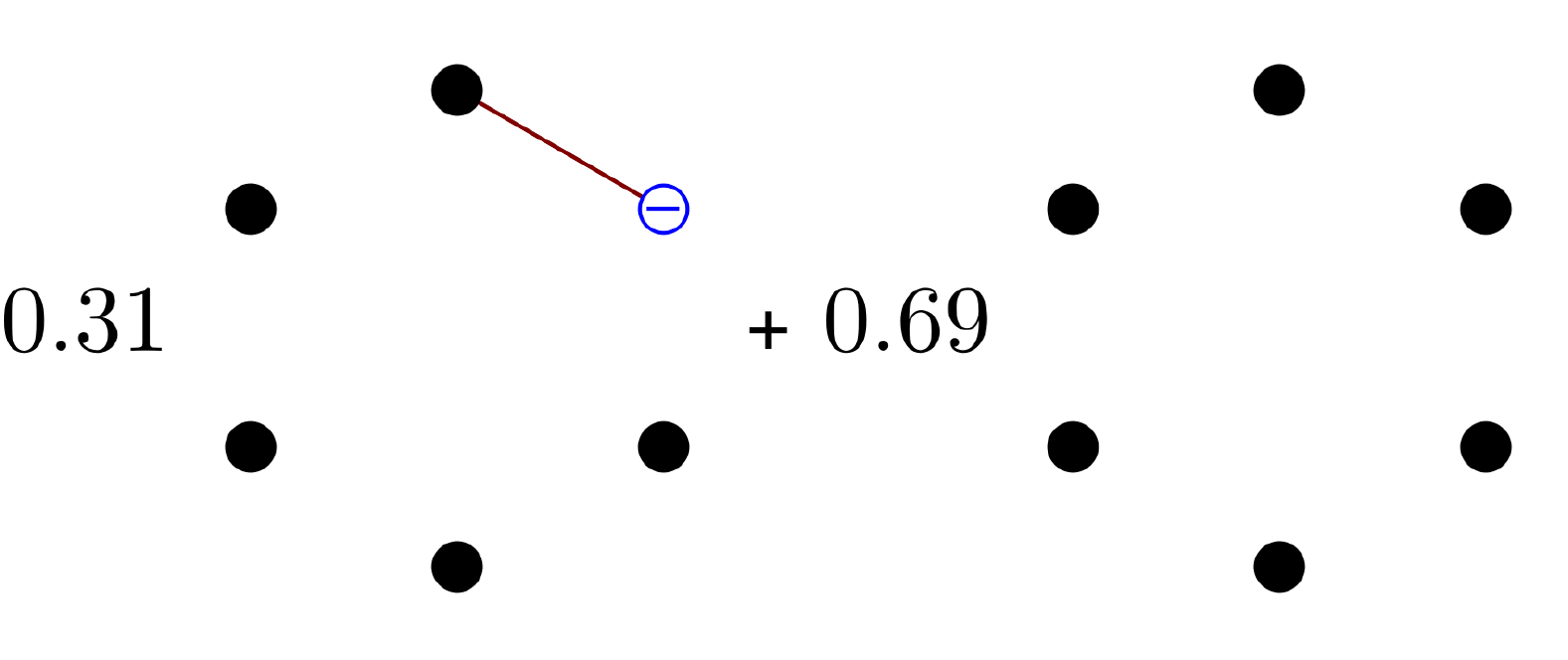}%
}\hspace{4em}
\subcaptionbox{$n=7$}{%
  \includegraphics[height=\myheight]{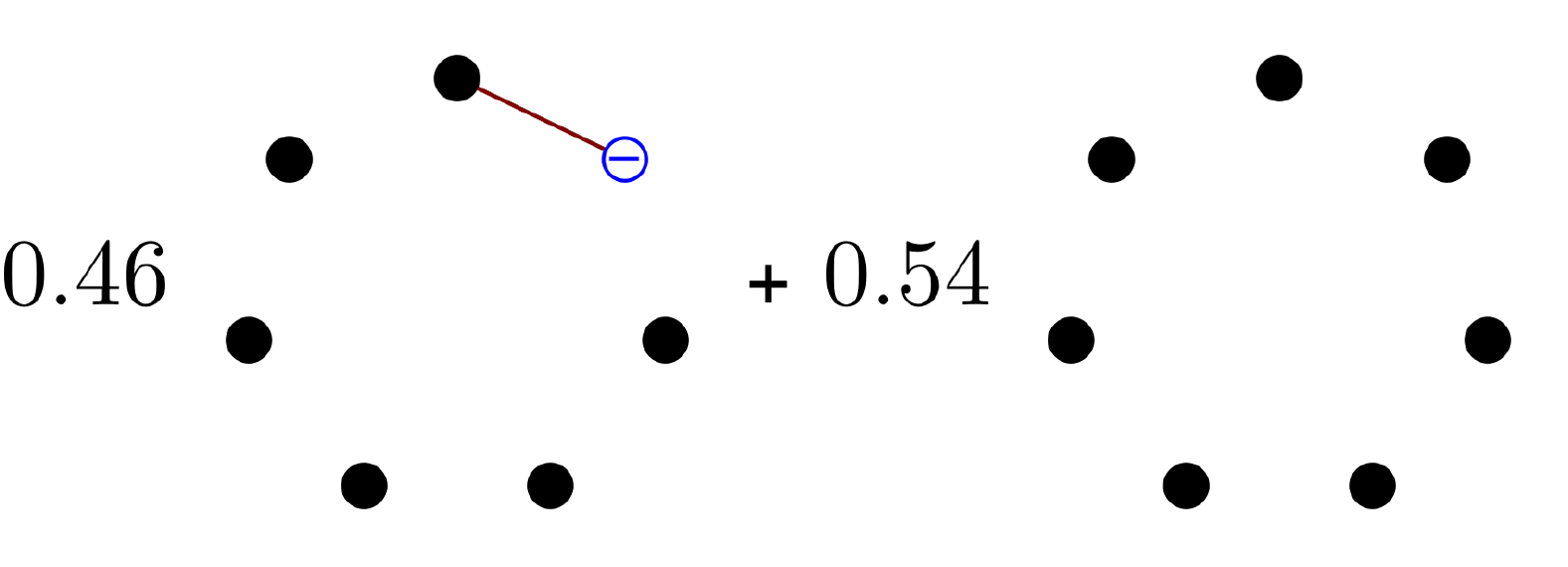}%
}\\[2em]
\subcaptionbox{$n=8$}{%
  \includegraphics[height=\myheight]{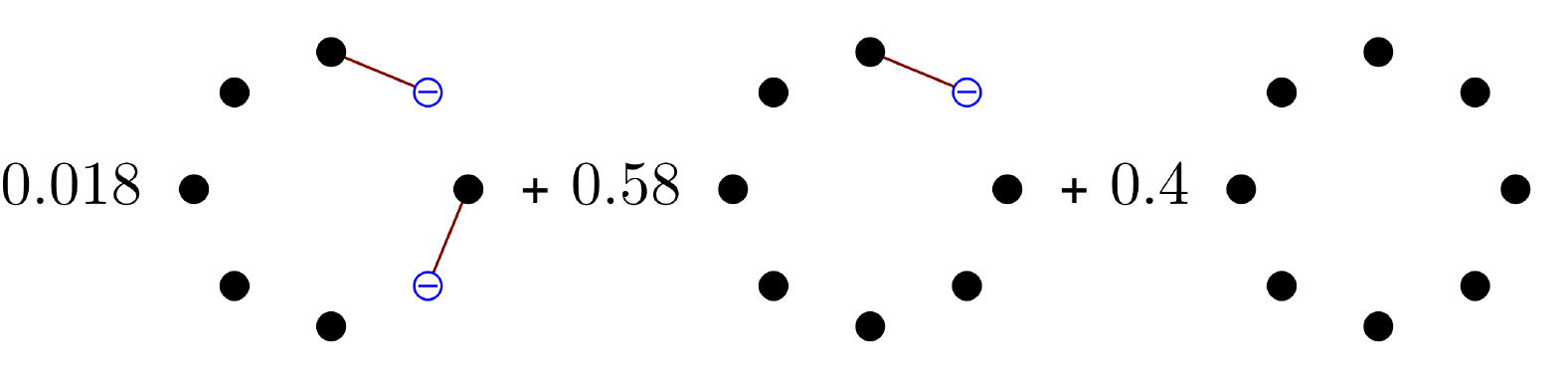}%
}

\caption{Positive contributions to the optimal affine combination for $\ket{H}^{\otimes n}$, written as convex combinations of stabiliser states. These states are represented as decorated graph states where hollow nodes indicate Hadamard action on the respective nodes and signs represent the respective sign of the stabiliser generator. Note that these states have only $\lfloor\log_2 n\rfloor$ contributions which themselves are products of $\ket +$ and Bell states.}
\label{fig:plus_states}
\end{figure}

\begin{figure}

\setlength{\myheight}{1.8cm}

\centering
\subcaptionbox{$n=2$}{%
  \includegraphics[height=1.1\myheight]{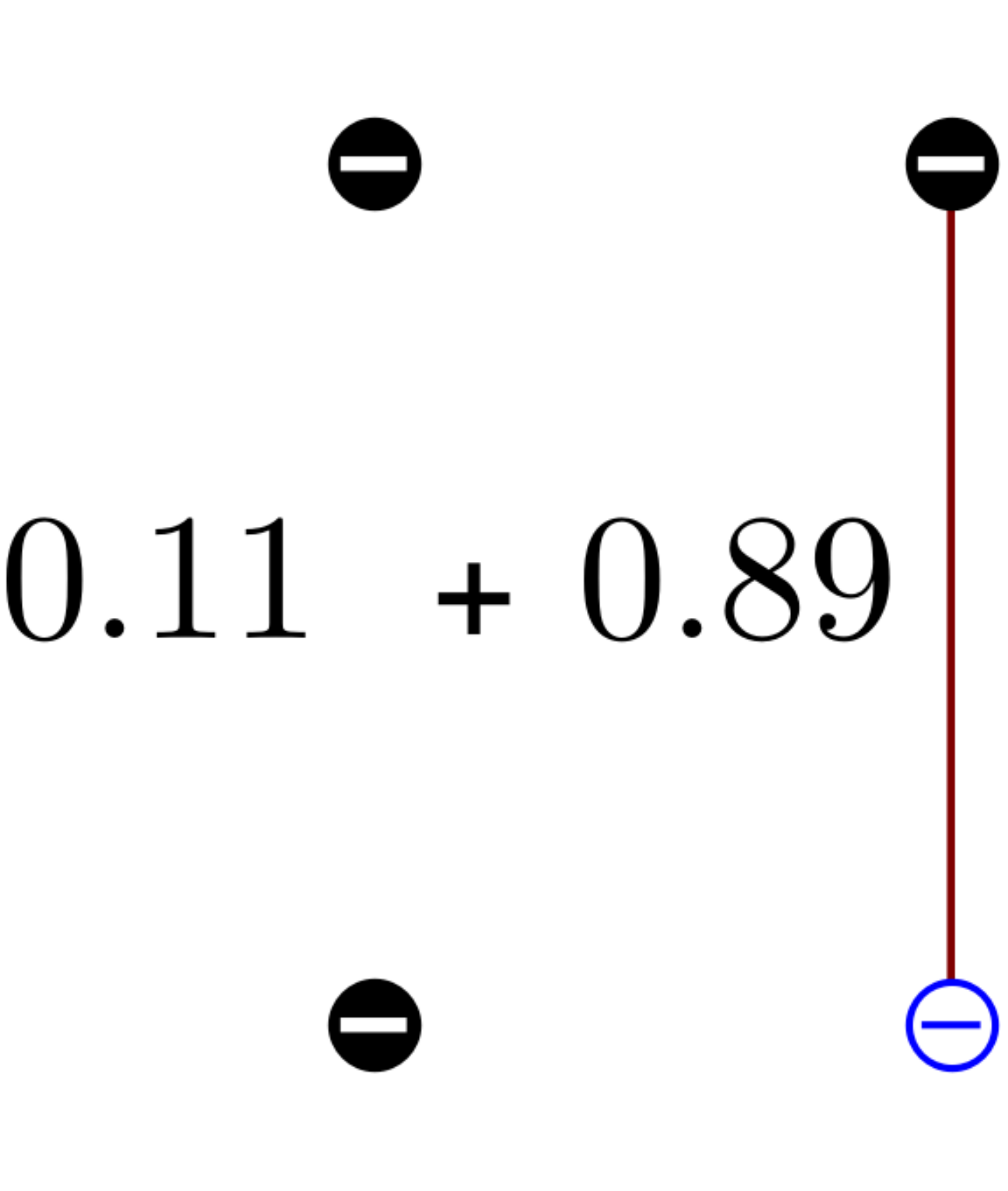}%
}\hspace{4em}
\subcaptionbox{$n=3$}{%
  \includegraphics[height=0.85\myheight]{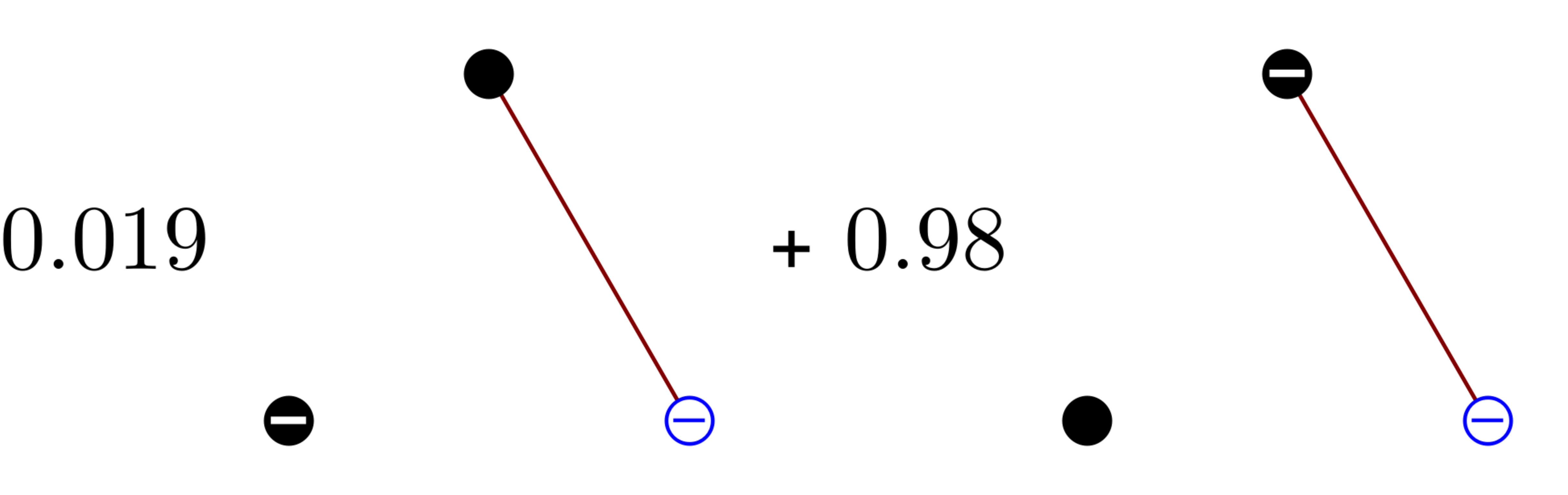}%
}\\[2em]
\subcaptionbox{$n=4$}{%
  \includegraphics[height=\myheight]{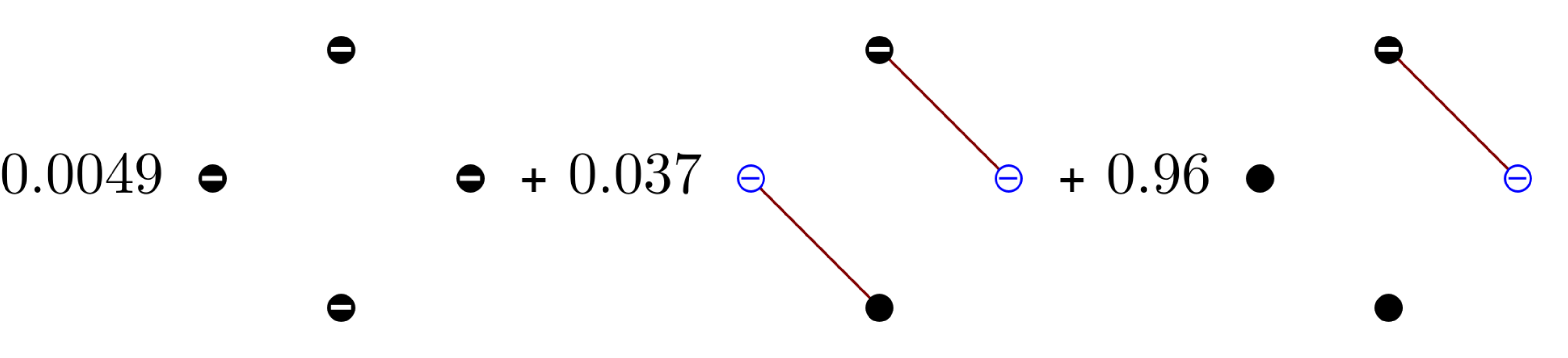}%
}\\[2em]
\subcaptionbox{$n=5$}{%
  \includegraphics[height=\myheight]{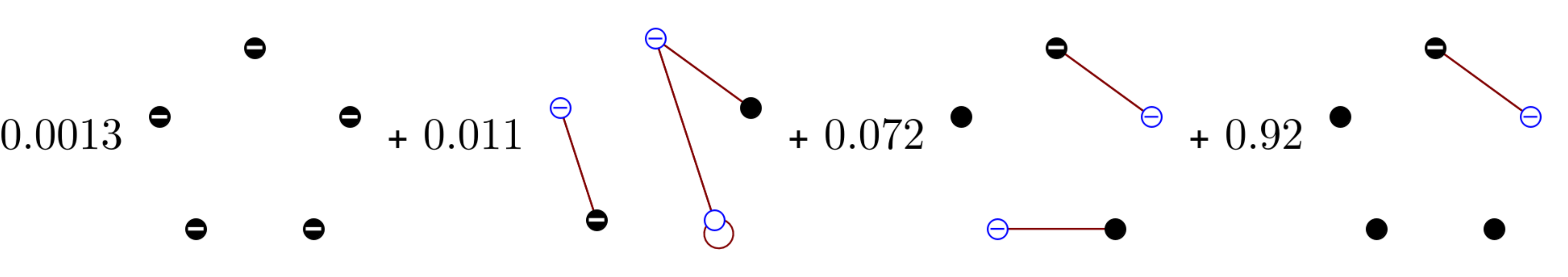}%
}\\[2em]
\subcaptionbox{$n=6$}{%
  \includegraphics[height=\myheight]{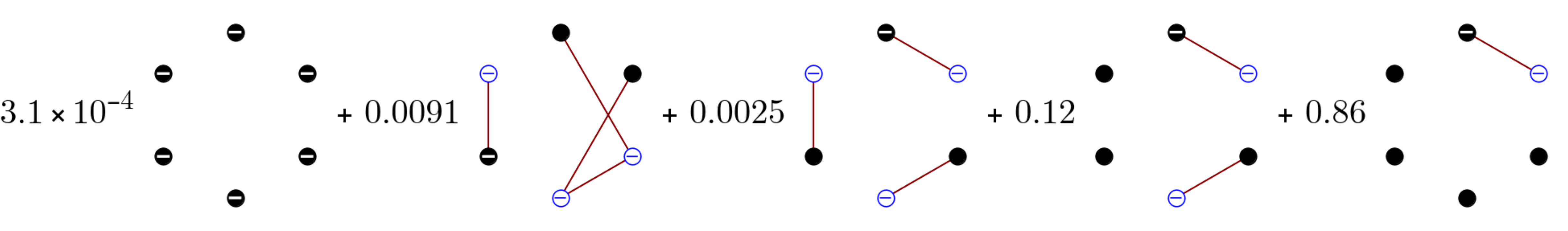}%
}\\[2em]
\subcaptionbox{$n=7$}{%
  \includegraphics[height=2.2\myheight]{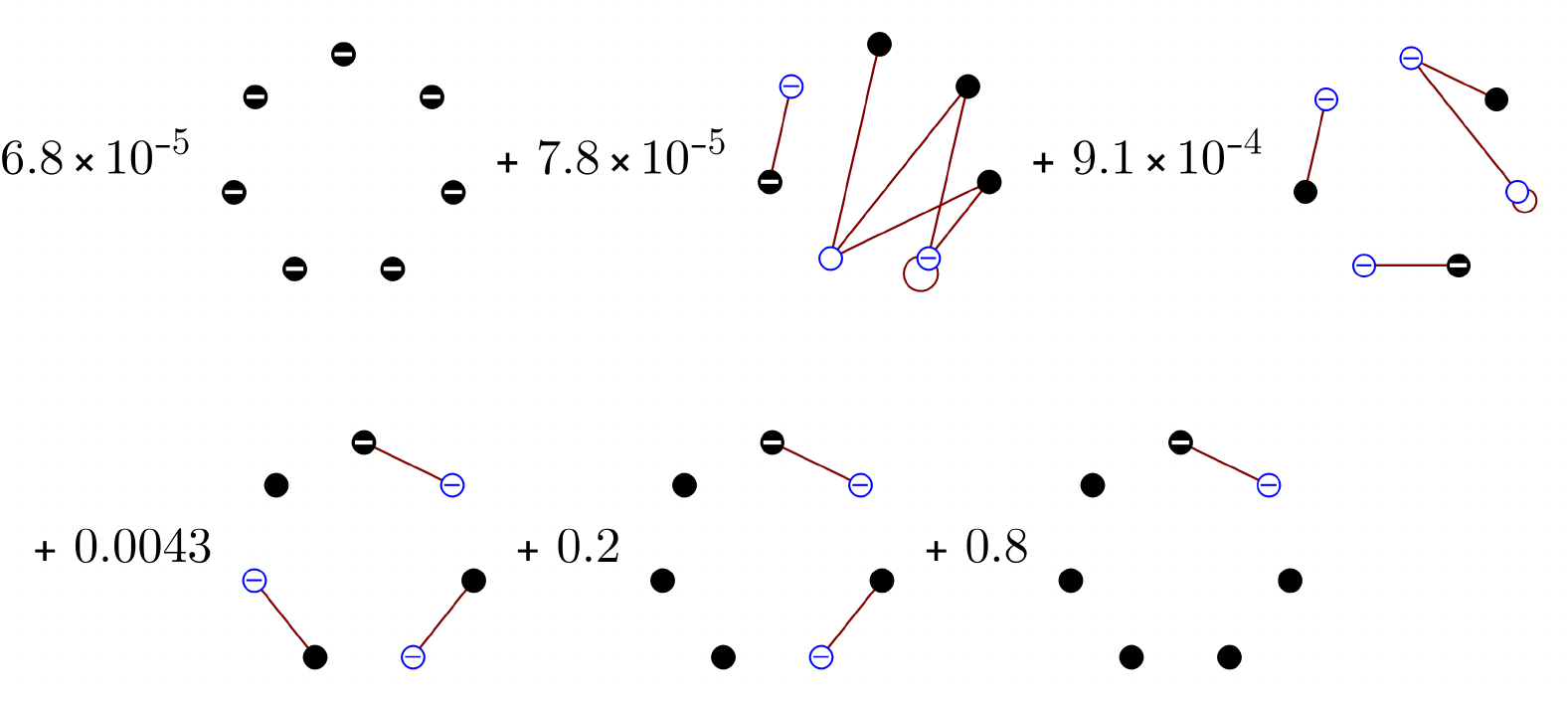}%
}\\[2em]
\subcaptionbox{$n=8$}{%
  \includegraphics[height=2.2\myheight]{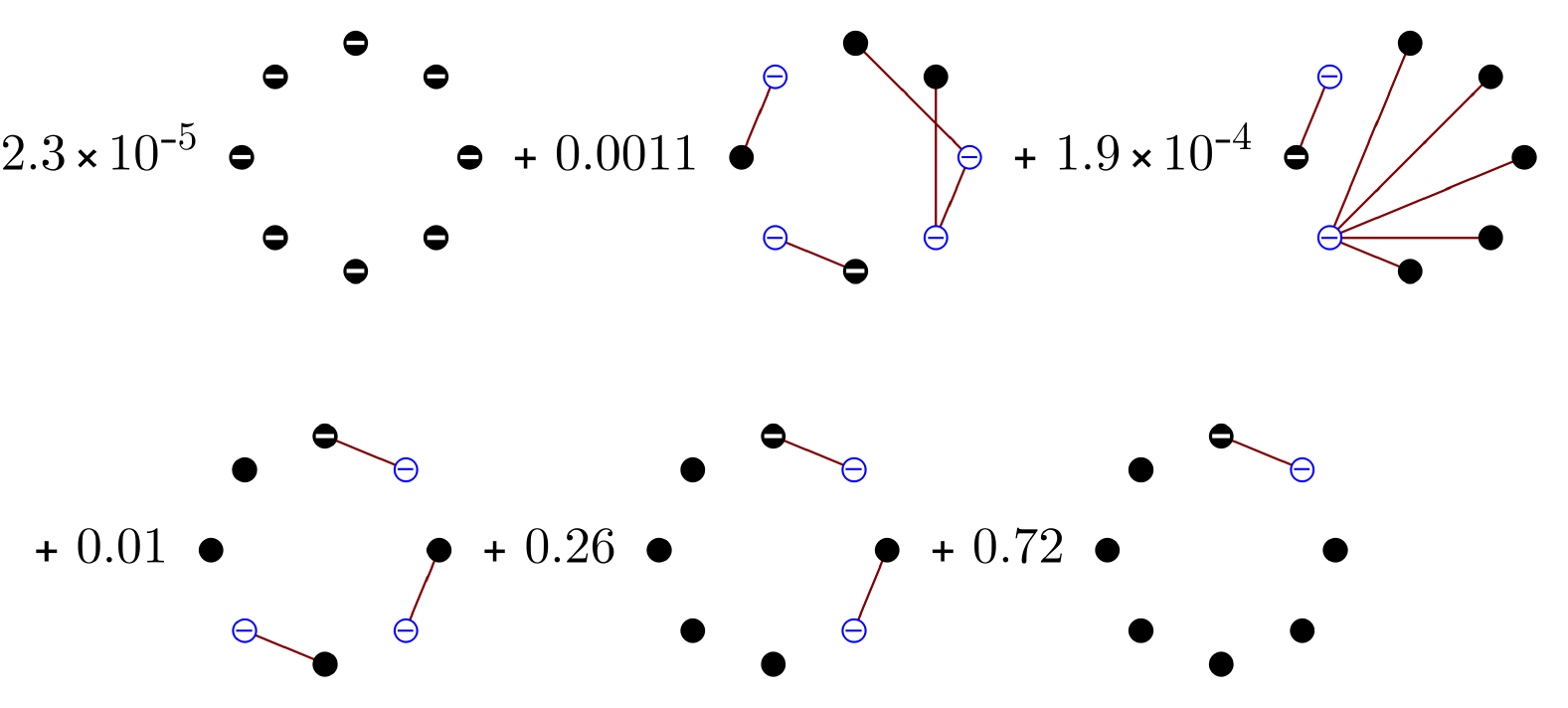}%
}

\caption{Negative contributions to the optimal affine combination for $\ket{H}^{\otimes n}$, written as convex combinations of stabiliser states. Note that these states have more contributions compared to the positive terms and seem partially irregular.}
\label{fig:minus_states}
\end{figure}

In contrast, the negative contributions $\sigma^-$, shown in Fig.~\ref{fig:minus_states}, have less structure and seem to be partially irregular. Of course, $\sigma^-$ has a non-unique convex combination and thus part of structure could be shadowed by the non-uniqueness. Nevertheless, since the dominant part of the contributions consists of products of $\ket{\pm}$ and the Bell states $\ket{\Psi^\pm}$, it is reasonable to assume that the $\sigma^-$ can be approximated by Bell states. We suspect that this approximation is quite good, at least for a moderate number of qubits, due to the apparent suppression of vertices with more complex structure.

Motivated by these observations, we define the following polytope:
\begin{equation}
 \ASP^H_n = \conv \Pi_H\left( \left\{ \text{all } n \text{ qubits states that are products of } \ket\pm \text{ and } \ket{\Psi^\pm} \right\} \right).
\end{equation}
By Eq.~\eqref{eq:wt_prod_states}, we can compute these states efficiently from the signed weight enumerators of $\ket\pm$ and $\ket{\Psi^\pm}$. Note that the projection of $\ket + \otimes \ket -$ is a convex combination of the projected Bell states and thus only states with ``all plus'' or ``all minus'' contributions are extremal in $\ASP^H_n$. Let $\mathcal W^H_n$ be the set of vertices of $\ASP^H_n$ and $m=\lfloor n/2 \rfloor$. We can explicitly enumerate its elements by tuples $(i,j,k)\in \{0,\dots,m\}^3$ such that $i+j+k=m$. Every such tuple corresponds to a product of $i$ $\ket \Psi^+$, $j$ $\ket \Psi^-$ and $(2k+n-2m)$ $\ket \pm$ states. Hence, the number of vertices is 
\begin{equation}
 K := |\mathcal W^H_n| = 2 \sum_{i=0}^m (m+1-i) = (m+1)(m+2).
\end{equation}
We define the approximate robustness of $\ket H^{\otimes n}$ as the robustness with respect to the polytope $\ASP^H_n$:
\begin{equation}
\label{eq:approx_robustness}
r^H_n := \min \left\{ \|x\|_1 \; \bigg| \; x\in\R^K: \; \rho = \sum_{i=1}^K x_i w_i \text{ with } w_i\in\mathcal W^H_n\right\}.
\end{equation}

Since the optimisation is over a subset of all projected stabiliser states, $r^H_n$ is an upper bound for $\RM(\ket H^{\otimes n})$. Moreover, it can be efficiently evaluated since both the complexity of computing $\mathcal{W}^H_n$ and of the $\ell_1$-minimisation is $O(n^4)$. Figure \ref{fig:approx_rom_plot} shows a comparison of $r^H_n$ with the exact robustness. From the previous analysis it is clear that the approximation is exact for $n\leq 4$. The deviation from the exact data for $4 < n \leq 9$ is at most 0.06\% and thus negligible. However, we expect that the deviation becomes larger the higher $n$ is, since it is likely that the importance of multipartite entangled contributions increases. Nevertheless, the approximation seems to be surprisingly good. The approximate data again follows an exponential increase with $n$, predicting an asymptotic regularised robustness of about $(1.2829\pm0.0017)$ which is compatible with the prediction $(1.283\pm 0.002)$ from the exact data.

However, this approach is limited to $n\leq 26$. For larger $n$, the $\ell_1$-minimisation lacks a feasible solution, which can only be the case if $\ket H^{\otimes n}$ is not in the affine span of the product states $\mathcal W^H_n$. This indicates that the dimension of the subpolytope $\ASP^H_n$ becomes too small. A solution to these infeasibility problems will be discussed in Sec.~\ref{sec:rom_hierarchy}.

\begin{figure}[h]
\begin{minipage}{0.68\textwidth}
 \input{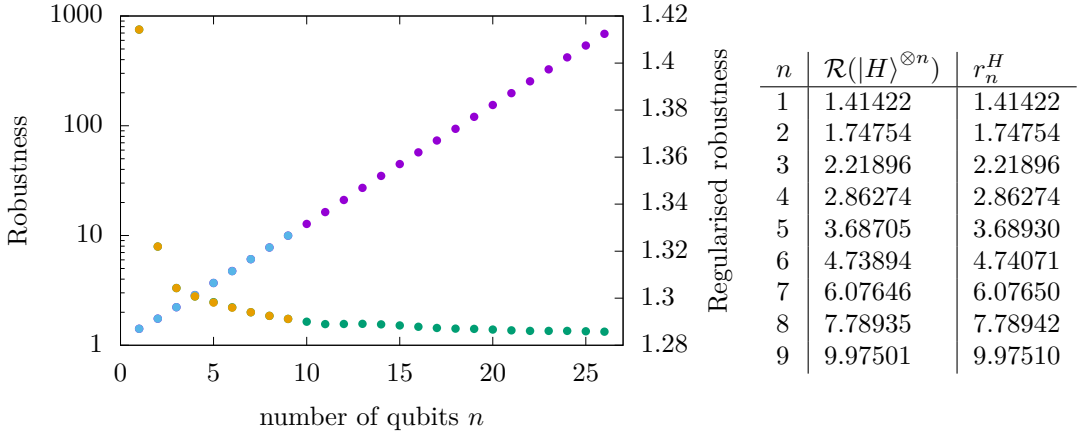}
\end{minipage}
\begin{minipage}{0.2\textwidth}
 \begin{tabular}{l|l|l}
   $n$ & $\mathcal{R}(\ket H^{\otimes n})$ & $r^H_n$ \\ \hline
   1 & 1.41422 & 1.41422 \\
   2 & 1.74754 & 1.74754 \\
   3 & 2.21896 & 2.21896 \\
   4 & 2.86274 & 2.86274 \\ 
   5 & 3.68705 & 3.68930 \\
   6 & 4.73894 & 4.74071 \\
   7 & 6.07646 & 6.07650 \\
   8 & 7.78935 & 7.78942 \\
   9 & 9.97501 & 9.97510
 \end{tabular}
\end{minipage}

\caption{Exact (blue, orange) and approximate (purple, green) robustness and regularised robustness of the magic state $\ket{H}^{\otimes n}$ as a function of the number of qubits $n$.}
\label{fig:approx_rom_plot}
\end{figure}

\paragraph{Optimal solutions for the $\ket T^{\otimes n}$ state} 

As in the previous case, the two connected vertices of the projected 2-qubit polytope constitute a dominant part in the optimal solutions. They are not projections of Bell states, so we will denote their representatives by $\ket{\gamma^\pm}$ and define them to be the states stabilised by $\{X_1Z_2,Z_1X_2\}$ and $\{-X_1Z_2,-Z_1Y_2\}$, respectively. The analysis of the optimal solutions shows that the $\sigma^+$ states are convex combinations of products of $\ket +$ and the maximally entangled state $\ket{\gamma^+}$. Moreover, they seem to be even more sparse than for the previous case, see Fig.~\ref{fig:Tplus_states}. As in the case of $\ket H$, the $\sigma^-$ state shows only partial structure, see Fig.~\ref{fig:Tminus_states} .

\begin{figure}
\setlength{\myheight}{1.8cm}

\centering
\subcaptionbox{$n=3$}{%
  \includegraphics[height=0.9\myheight]{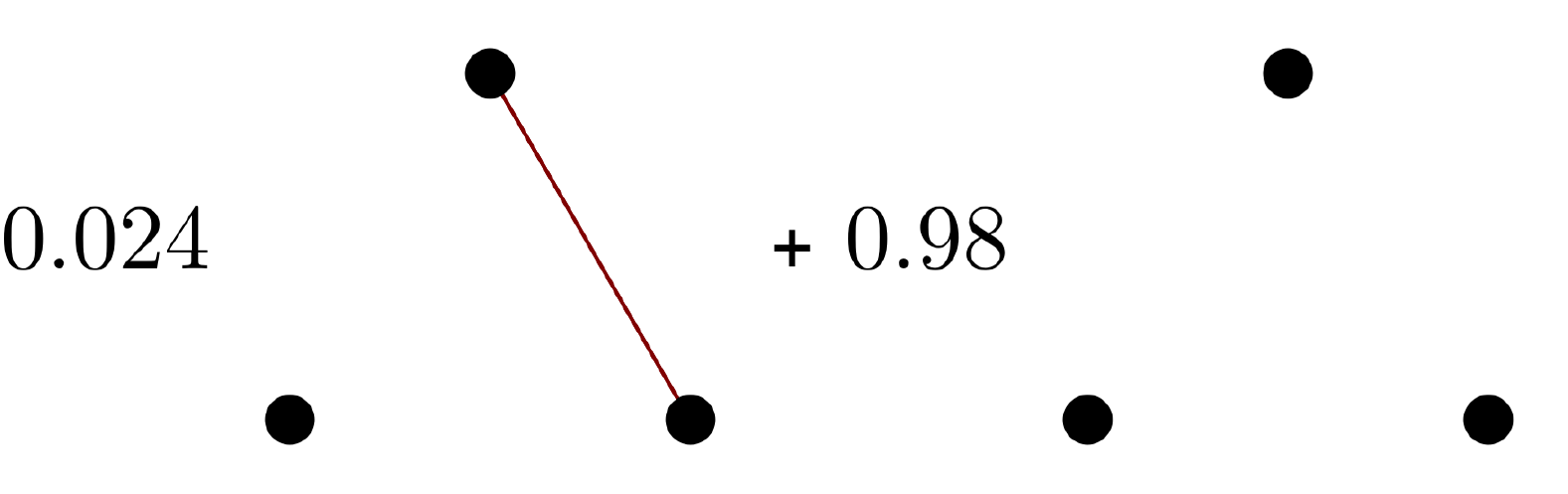}%
}\hspace{4em}
\subcaptionbox{$n=4$}{%
  \includegraphics[height=\myheight]{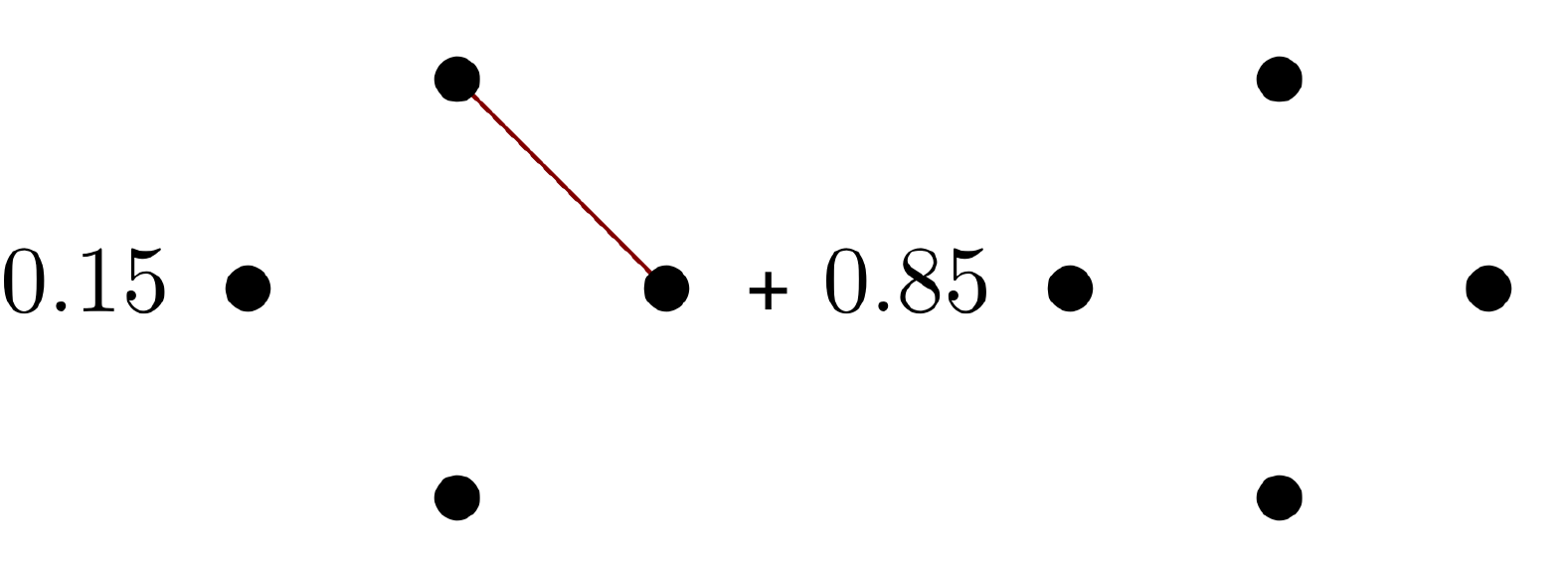}%
}\\[2em]
\subcaptionbox{$n=5$}{%
  \includegraphics[height=\myheight]{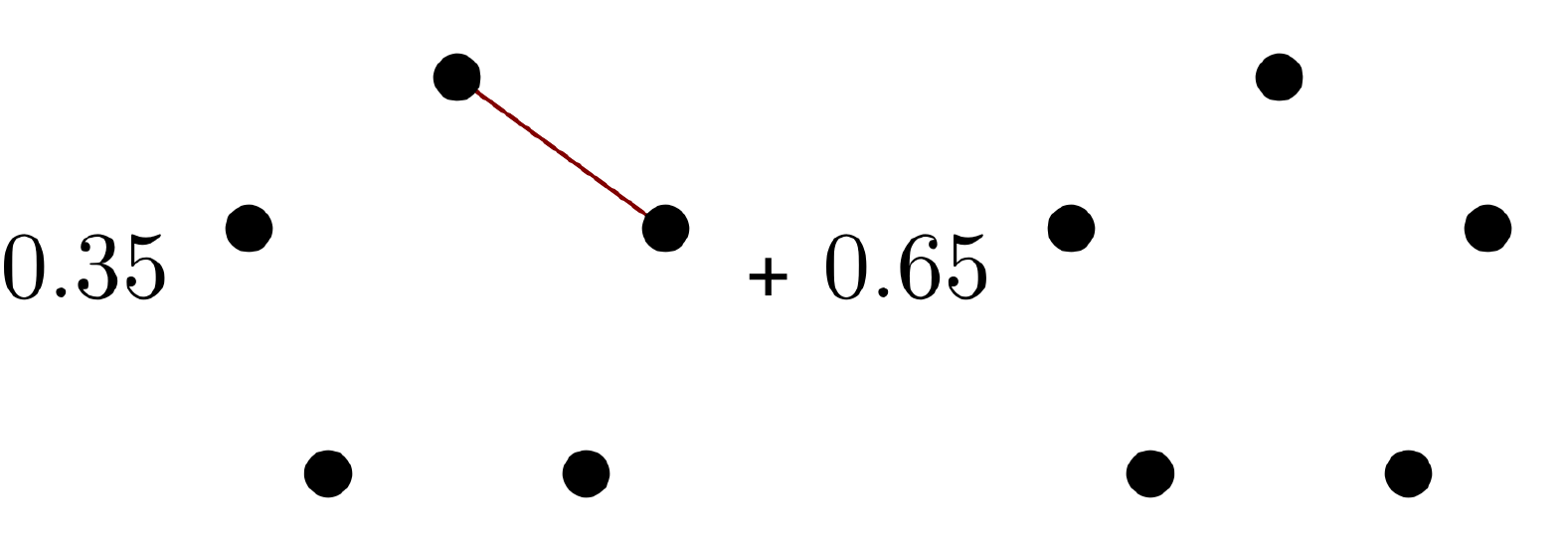}%
}\hspace{4em}
\subcaptionbox{$n=6$}{%
  \includegraphics[height=\myheight]{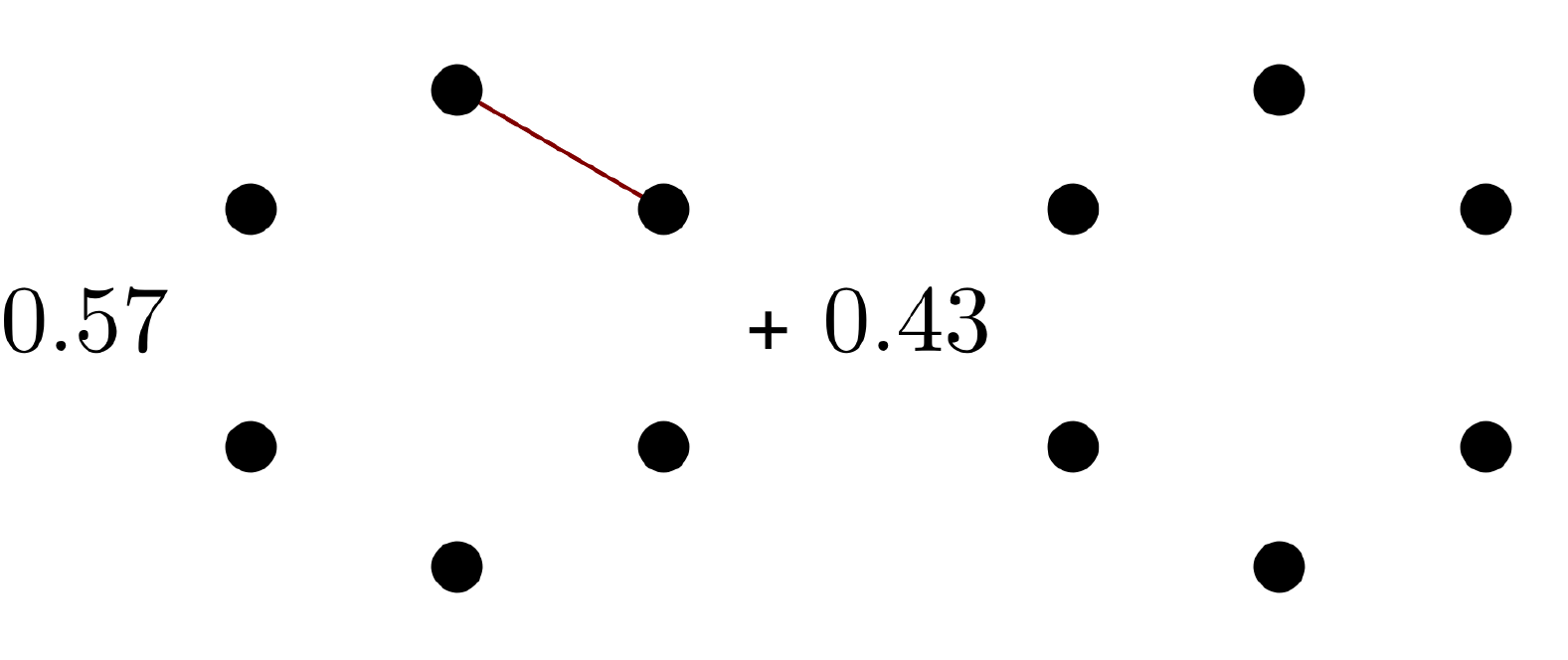}%
}\\[2em]
\subcaptionbox{$n=7$}{%
  \includegraphics[height=\myheight]{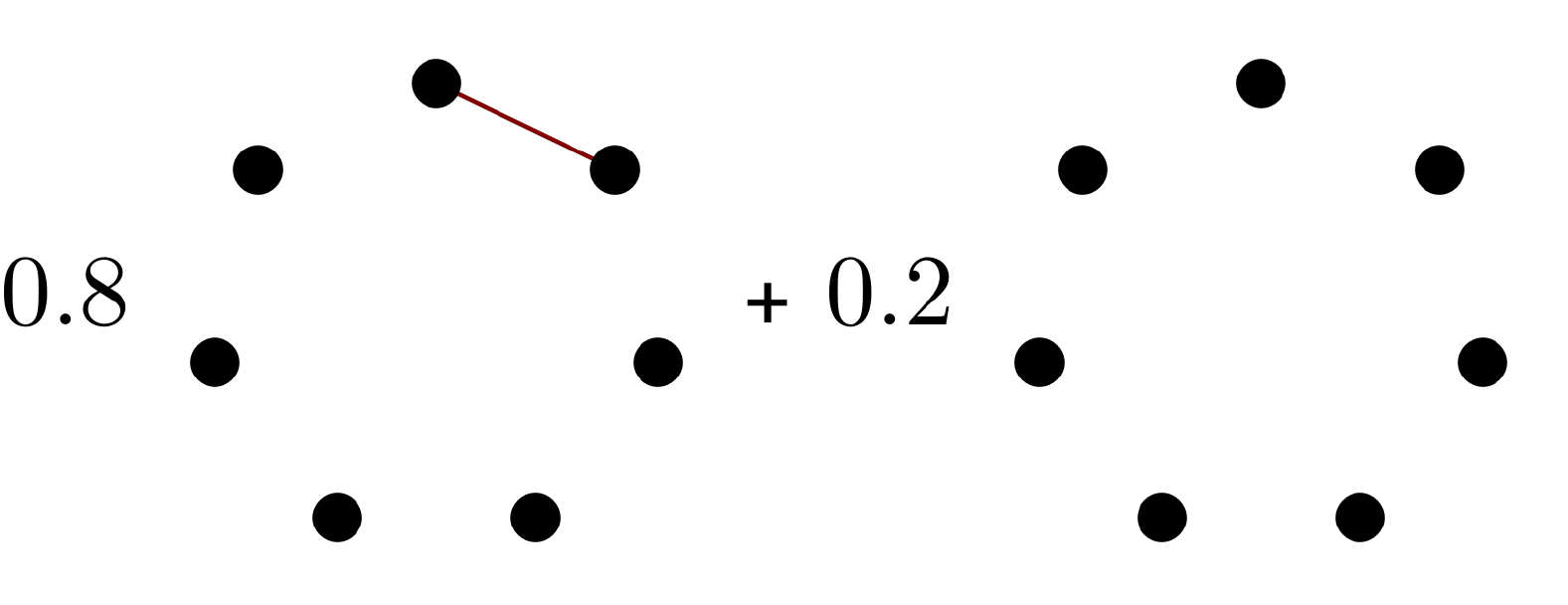}%
}\hspace{4em}
\subcaptionbox{$n=8$}{%
  \includegraphics[height=\myheight]{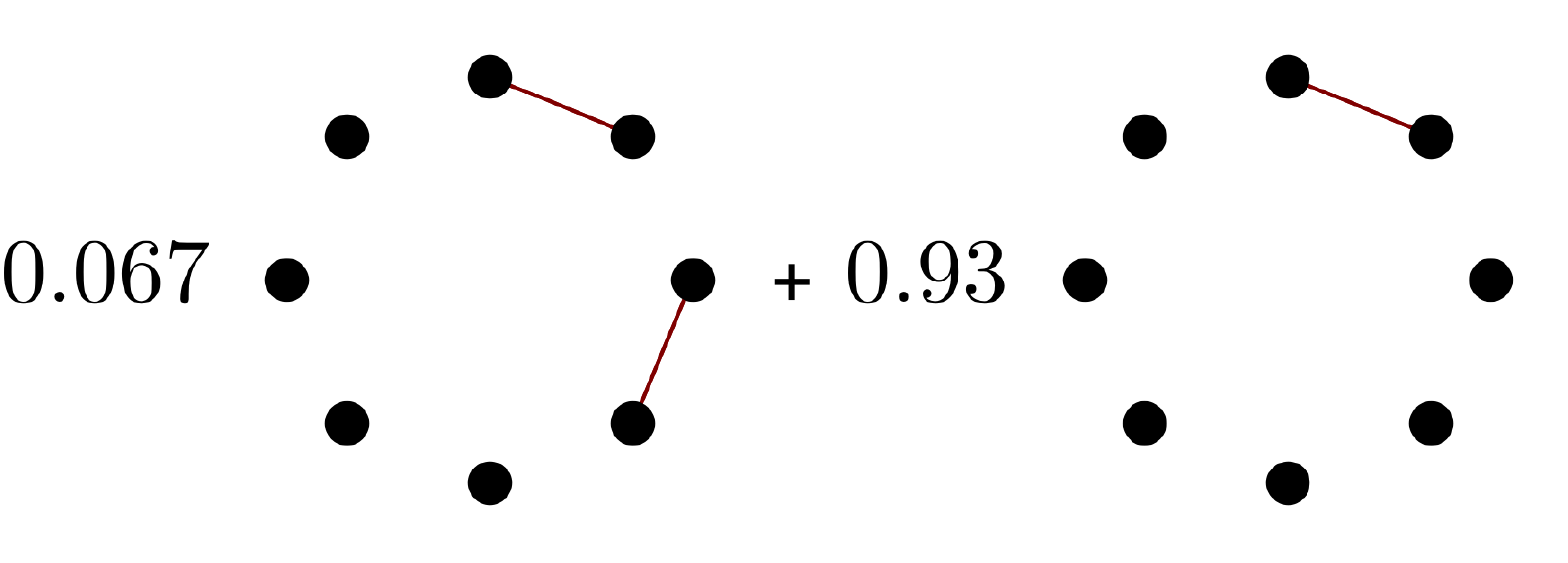}%
}

\caption{Positive contributions to the optimal affine combination for $\ket{T}^{\otimes n}$ and $3 \leq n \leq 8$, written as convex combinations of stabiliser states. These states are again represented as decorated graph states, see Sec.~\ref{sec:computations}. Note that these states have less than $\lfloor\log_2 n\rfloor$ contributions which themselves are product states made from $\ket +$ and Bell states.}
\label{fig:Tplus_states}
\end{figure}

\begin{figure}
\begin{minipage}{0.68\textwidth}
 \input{gfx/i_approx_romplotT}
\end{minipage}
\begin{minipage}{0.2\textwidth}
 \begin{tabular}{l|l|l}
   $n$ & $\mathcal{R}(\ket T^{\otimes n})$ & $r^T_n$ \\ \hline
    1 & 1.73206 & 1.73206 \\
    2 &	2.23206 & 2.23206 \\
    3 &	3.09808 & 3.09808 \\
    4 &	4.33100 & 4.33316 \\
    5 & 6.04494 & 6.04494 \\
    6 &	8.35898 & 8.36006 \\
    7 &	11.5114 & 11.5117 \\
    8 &	15.8436 & 15.8492 \\
    9 & 22.1823 & 22.2499 \\
    10 & 30.7056 & 30.9411
 \end{tabular}
\end{minipage}
\caption{Exact (blue, orange) and approximate (purple, green) robustness and regularised robustness of the magic state $\ket{T}^{\otimes n}$ as a function of the number of qubits $n$.}
\label{fig:approx_rom_plot_T}
\end{figure}

\begin{figure}
\setlength{\myheight}{1.8cm}

\centering
\subcaptionbox{$n=3$}{%
  \includegraphics[height=0.9\myheight]{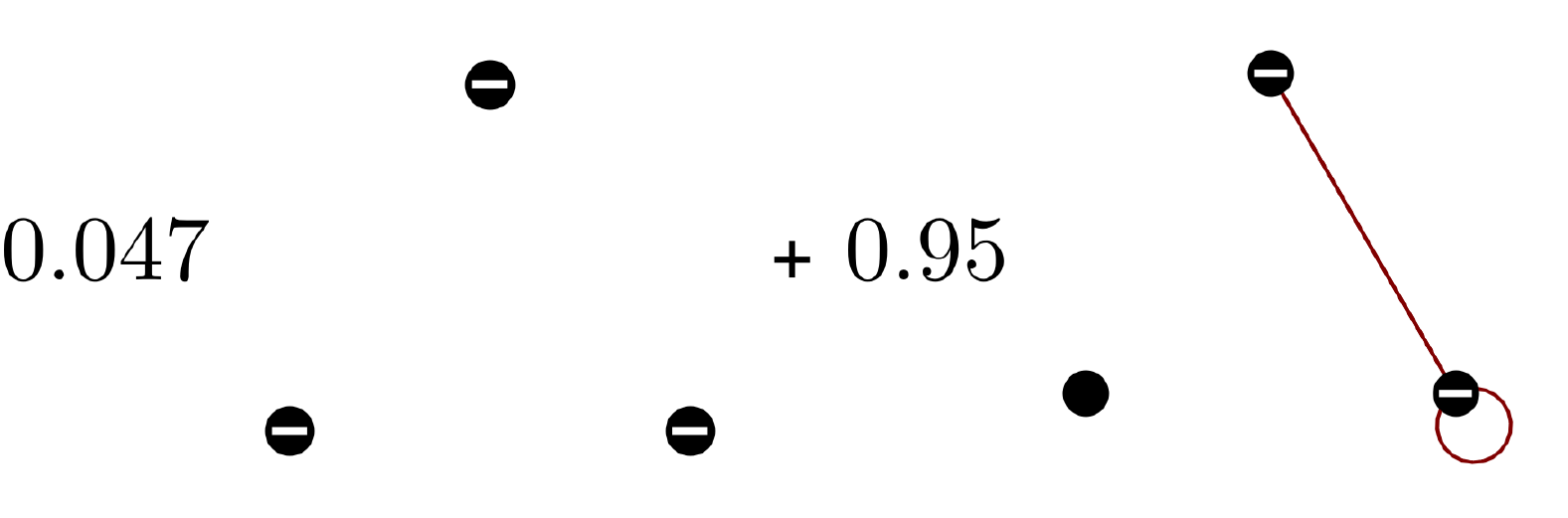}%
}\hspace{4em}
\subcaptionbox{$n=4$}{%
  \includegraphics[height=\myheight]{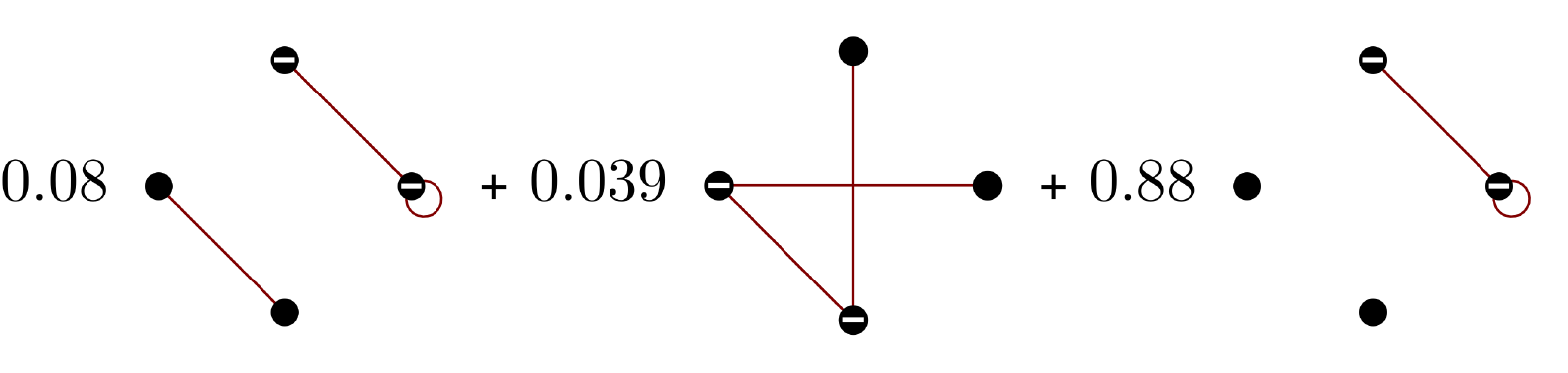}%
}\\[2em]
\subcaptionbox{$n=5$}{%
  \includegraphics[height=\myheight]{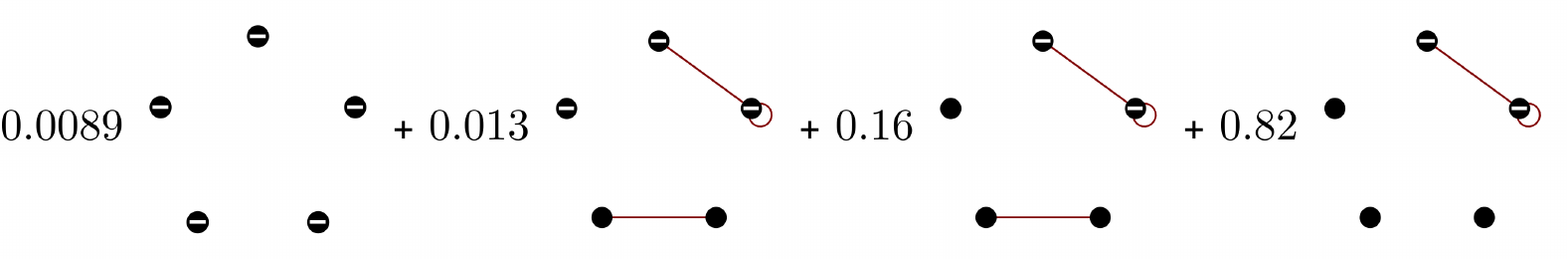}%
}\hspace{4em}
\subcaptionbox{$n=6$}{%
  \includegraphics[height=\myheight]{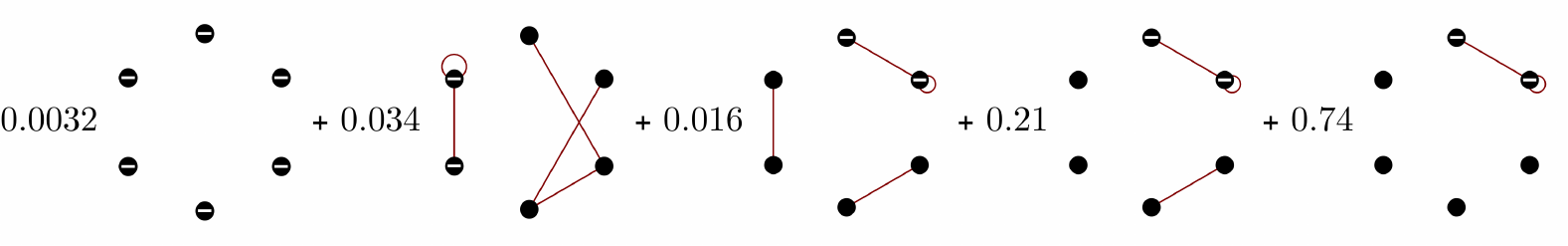}%
}\\[2em]
\subcaptionbox{$n=7$}{%
  \includegraphics[height=2.2\myheight]{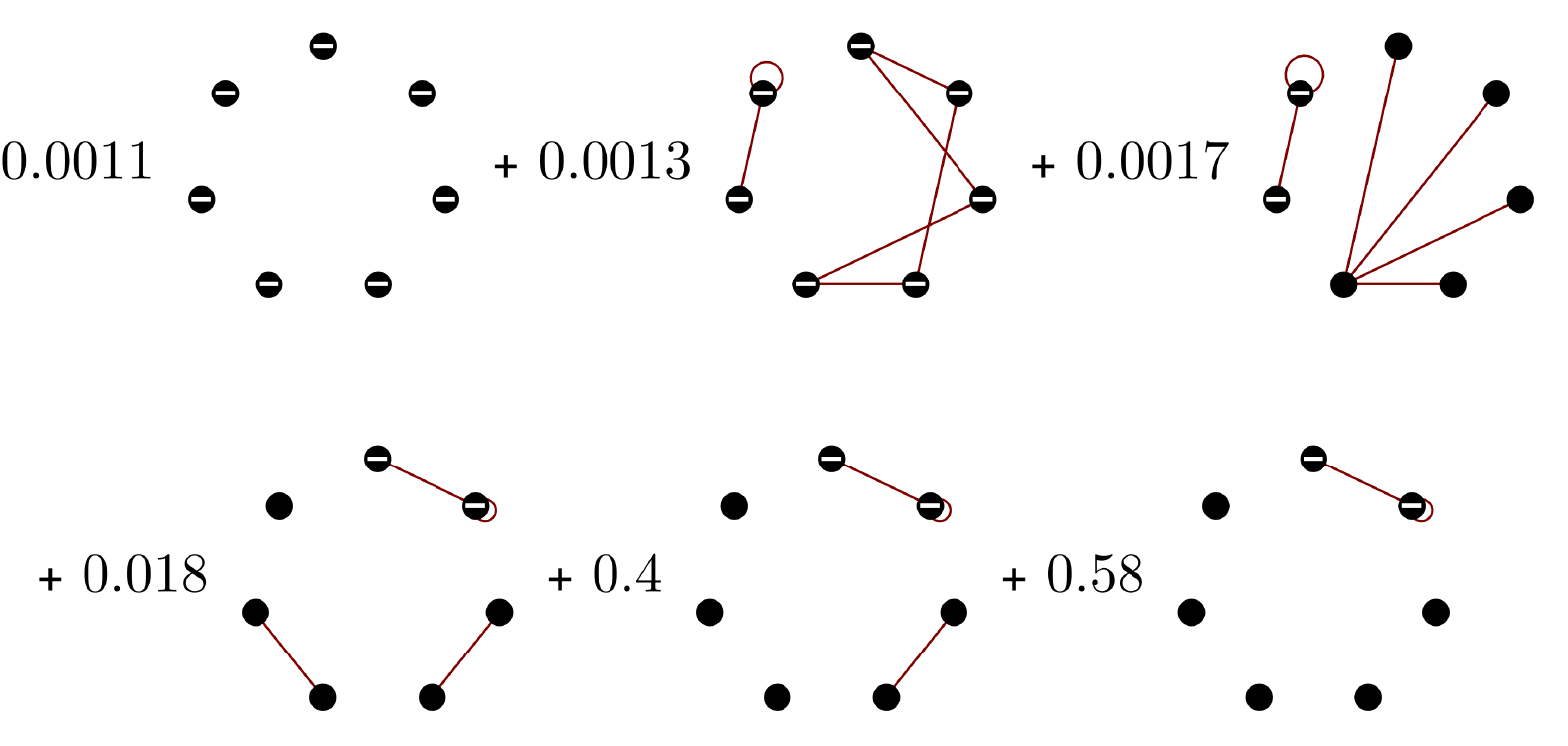}%
}\\[2em]
\subcaptionbox{$n=8$}{%
  \includegraphics[height=2.2\myheight]{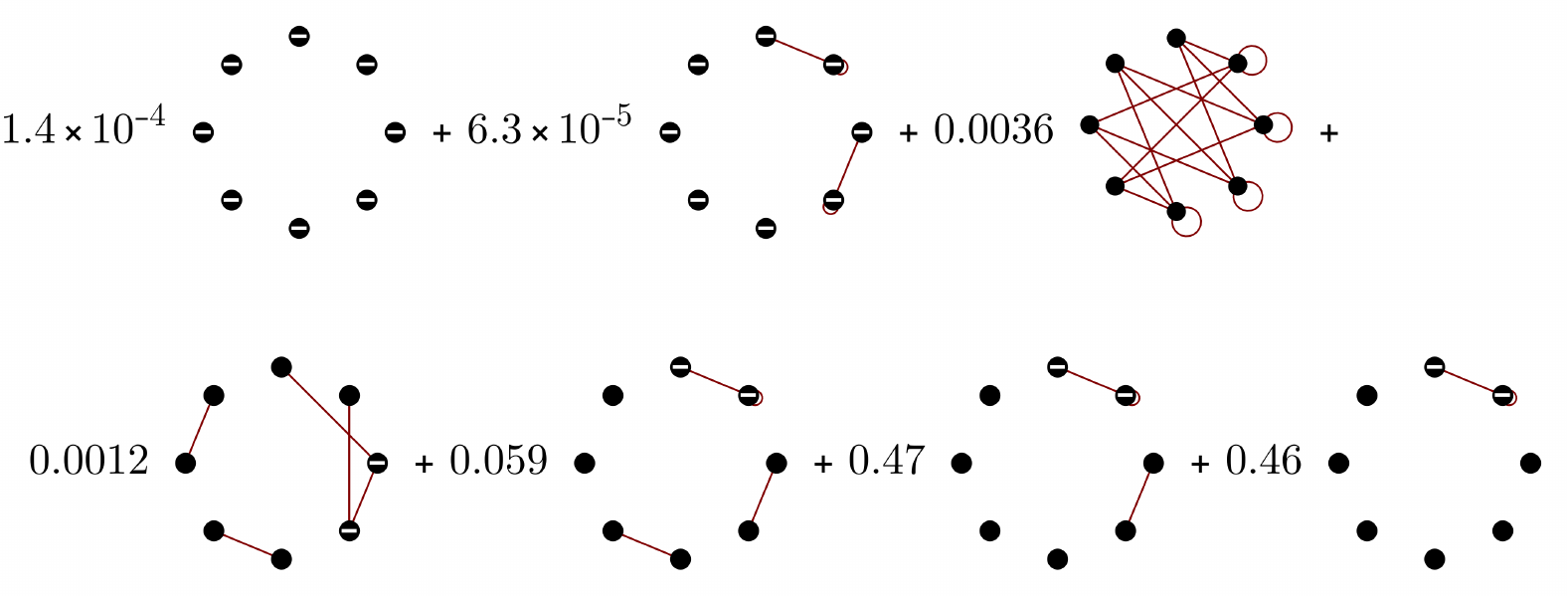}%
}

\caption{Negative contributions to the optimal affine combination for $\ket{T}^{\otimes n}$ and $3\leq n \leq 8$, written as convex combinations of stabiliser states. These states are again represented as decorated graph states, see Sec.~\ref{sec:computations}.}
\label{fig:Tminus_states}
\end{figure}

The similarities suggest that the robustness for $\ket T^{\otimes n}$ can be well approximated using a similar procedure as in the last section. To this end, we define the polytope
\begin{equation}
 \ASP^T_n = \conv \Pi_T\left( \left\{ \text{all } n \text{ qubits states that are products of } \ket{\pm} \text{ and } \ket{\gamma^\pm}  \right\} \right).
\end{equation}
The approximate robustness $r^T_n$ is again defined with respect to this polytope. The vertices $\mathcal W^T_n$ can be efficiently computed using the same procedure as in the $\ket H$ case and the approximation is exact for $n \leq 3$. Figure \ref{fig:approx_rom_plot_T} shows the approximate robustness compared to the exact results. The approximation is again surprisingly good with a maximum deviation from the exact data of around 0.8\%. Although this error is still small, it is an order of magnitude larger than for the $\ket H$ state. The approximation yields an asymptotic regularised robustness of $(1.3916\pm 0.0014)$ which is slightly larger than the result from the exact data. Similar to the last section, the applicability of this approximation is limited to $n\leq 24$ due to the infeasibility of the optimisation problem for larger $n$. In the next section, we will show how to generalise this approximation to overcome the feasibility problems.

\subsection{Finite hierarchy of RoM approximations}
\label{sec:rom_hierarchy}

In general, the idea of restricting to at most $k$-partite entangled stabiliser states leads to a hierarchy of approximations with levels $1\leq k \leq n$. Clearly, for $k=n$ the exact problem is recovered. The set of at most $k$-partite entangled $n$-qubit stabiliser states can be constructed by taking all possible tensor products of states in $\stab(i)$ for $1\leq i\leq k$ which result in $n$-qubit states. However, without the presence of additional symmetries, this will still result in an exponentially large set since already the set of fully separable stabiliser states ($k=1$) has size $6^n$.

Hence, we assume that we want to compute approximations to $\RM(\rho)$ where $\rho$ is a symmetric $n$-qubit state (not necessarily pure) such that the stabiliser symmetry group contains at least the symmetric group $S_n$. In particular, this applies to the magic states $\ket H^{\otimes n}$ and $\ket T^{\otimes n}$. In this case, we are able to give $\mathrm{poly}(n)$ upper bounds on the runtime for every fixed level $k < n$.

Following Lemma \ref{lem:proj_prod_states} and Section \ref{sec:id_symmetries}, the set of $S_n$-projections of $k$-partite entangled $n$-qubit stabiliser states can be constructed from the vertices of the projected polytopes $\overline{\SP}_i=\Sym(\SP_i)$ for $1\leq i\leq k$ which have \emph{fully entangled} representatives. Let us denote the sets of representatives by $V_i\subset \stab(i)$. Since the order does not matter, the possible ways to take tensor products of these sets are exactly captured by (descending) partitions of $n$ into parts with size at most $k$. We will denote such a partition by $\lambda\vdash_k n$. Then, we define the subpolytope of projected $k$-partite entangled states as
\begin{equation}
\label{eq:stab_subpolytope}
 Q_{n,k} := \conv \Sym\left( \bigcup_{\lambda\vdash_k n} \bigotimes_{i\in\lambda} V_i \right),
\end{equation}
and the \emph{$k$-th level of the RoM hierarchy} by the relaxation of Prob.~\ref{prob:rom} to the subpolytope $Q_{n,k}$. Clearly, this defines an upper bound $r_{n,k}(\rho)$ to the exact RoM $\RM(\rho)$.

To bound the runtime of the $k$-th level of the hierarchy, we have to count the vertices $\mathcal{W}_{n,k}$ of $Q_{n,k}$. An upper bound to this number is given by the number of tensor products appearing in Eq.~\eqref{eq:stab_subpolytope} up to permutations. Thus, let $\lambda$ be a (descending) partition of $n$ into $r$ parts, with no part larger than $k$:
\begin{equation}
 n = \lambda_1 + \dots + \lambda_r, \qquad k\geq \lambda_1 \geq \lambda_2 \geq \dots \geq \lambda_r > 0.
\end{equation}
This can be rewritten as
\begin{equation}
 n = \sum_{i=1}^k m_i \, i,
\end{equation}
where $0\leq m_i \leq n$ is the \emph{multiplicity} of $i$ in the partition $\lambda$. Since the permutations of the partition itself were already considered, the number of product states corresponding to the partition $\lambda$ is given, up to permutations, by
\begin{equation}
 \prod_{i=1}^k \left| V_i^{m_i}/S_{m_i} \right| = \prod_{i=1}^k \binom{m_i + L_i - 1}{L_i - 1}, \qquad L_i := |V_i|.
\end{equation}
Using that the number of fully entangled vertices is increasing with $i$, we can bound this number by
\begin{equation}
  \prod_{i=1}^k \binom{m_i + L_i - 1}{L_i - 1} \leq \prod_{i=1}^k m_i^{L_i} \leq n^{k\, L_k}.
\end{equation}
Finally, the number of partitions of $n$ with parts no greater than $k$ coincides with the number of partitions of $n$ into at most $k$ parts and is denoted by $p_k(n)$. A standard result in number theory is that
\begin{equation}
 p_k(n) = \frac{n^{k-1}}{k!(k-1)!} + O(n^{k-2}).
\end{equation}
Thus, we can bound the number of vertices  $\mathcal{W}_{n,k}$ to be
\begin{equation}
 |\mathcal{W}_{n,k}| \leq p_k(n)\, n^{k\, L_k} = \frac{n^{k\,(L_k+1)-1}}{k!(k-1)!} + O(n^{k\,(L_k+1)-2}).
\end{equation}
Since the dimension is $O(n^3)$, this implies that the runtime of the relaxation of Problem \ref{prob:rom} is polynomial in $n$ for a fixed $k$.

Finally, we remark that one has to know the vertex sets $V_i$ up to $k$ to run the $k$-th level of the hierarchy. Moreover, the bounds are very loose due to the fact we have not strictly bound the number of fully entangled vertices $L_i$ which is beyond the scope of this paper. However, by using the actual numbers for $L_i$, one can obtain much better bounds on $|\mathcal{W}_{n,k}|$ by evaluating the binomial coefficients. Let us illustrate this for the case of $\ket H^{\otimes n}$ and $k=2,3$: Using that $p_2(n)=\lfloor \frac{n}{2}\rfloor+1$, 
$p_3(n)=\lfloor\frac{(n+3)^3}{12}+\frac{1}{2}\rfloor$ and $L_1=L_2=L_3=2$, we find 
\begin{align}
 |\mathcal{W}_{n,2}| &\leq \left( \left\lfloor \frac{n}{2}\right\rfloor+1 \right) \binom{n+2-1}{2-1}^2 = O(n^3), \\
 |\mathcal{W}_{n,3}| &\leq \left\lfloor \frac{(n+3)^3}{12}+\frac{1}{2}\right\rfloor n^3 = O(n^6).
\end{align}
Note that we derived $|\mathcal{W}_{n,2}|=O(n^2)$ in the previous section using further information about the extremality of products.

\section{Conclusion \& Outlook}
\label{sec:conclusion}

In this work, we have studied the symmetries of the $n$-qubit stabiliser polytope and showed how to use these to greatly reduce the combinatorical complexity of computing the robustness of single-qubit magic states and to gain insight into the structure of the problem. 

We have determined the symmetry groups for the two types of single-qubit magic states and have constructed explicit stabiliser state representatives of the symmetry orbits. This has allowed us to evaluate the robustness of $\ket H^{\otimes n}$ for $n\leq 9$ and $\ket T^{\otimes n}$ for $n\leq 10$ qubits. Using the structure of the solutions, we have proposed an approximation based on at most bipartite entangled states which is efficient in $n$ and gives an upper bound on the exact robustness. Furthermore, the agreement with the exact data for $n\leq 10$ qubits is excellent. Since the RoM becomes effectively multiplicative for larger $n$, we expect that the approximation is still very good in the regime $n > 10$. Moreover, by restricting to $k$-partite entangled stabiliser states, we obtained a finite hierarchy of approximations which recovers the exact RoM for $k=n$. We showed that a fixed level $k < n$ of the hierarchy can be computed in $\poly(n)$ time.

We feel that the most interesting task left open in this work is to explain why even two-body entangled states are sufficient to produce excellent bounds on the RoM. 
This may be insightful in a wider context.
Indeed, sub-additivity of resource costs occurs in several areas of quantum information theory, most famously for the entanglement of formation \cite{hastings2009superadditivity}.
The violations to additivity in \cite{hastings2009superadditivity} can be proven to exist for randomised constructions in high dimensions.
This makes it hard to study the structure of the optimal solutions, or their behavior in a limit of many copies.
The combinatorial nature of the stabiliser polytope, and the observation that only few-body entanglement is enough to find almost-optimal solutions, suggest that RoM may provide an instance where understanding submultiplicativity is feasible.

\section{Acknowledgements}

We thank 
Earl Campbell, 
Mateus Ara\'ujo,
Felipe M.~Mora,
Felix Huber,
Frank Vallentin,
Arne Heimendahl,
and
Huangjun Zhu
for helpful discussions and comments. In particular, we want to thank Richard Kueng for productive discussions concerning the dual problem.

This work has been supported by the Excellence Initiative of the German Federal and State Governments (Grant ZUK 81), 
the German Research Foundation (DFG project B01 of CRC~183), and
ARO under contract W911NF-14-1-0098 (Quantum Characterization, Verification, and Validation).

\clearpage
\appendix

\section{Equivalence of the two robustness measures}
\label{sec:l1proof}

The equivalence given in Eq.~(\ref{eq:l1_tot_robustness}) is stated implicitly in \cite{howard_application_2017}. 
Here, we give an explicit proof.

\textcite{vidal_robustness_1999} defined the so-called \emph{total robustness} which is given by
\begin{equation}
 \label{eeq:tot_robustness}
 R(a) := \inf_{b \in S}R(a||b).
\end{equation}
For $S$ being a (compact) polytope, this can be rewritten as follows. Since $S$ is compact, the minimum $b^*$ is attained. Hence, $R(a) = R(a||b^*) =: s^*$ and
\begin{equation}
 b^+ := \frac{1}{1+s^*}\left( b^* + s^* a \right) \in S.
\end{equation}
Let $\{v_1,\dots,v_N\}$ be the vertices of $S$ and write $b^+,b^*\in S$ as convex combinations with coefficients $\lambda_i$ and $\mu_i$. It follows:
\begin{equation}
 a = (1+s^*)b^+ - s^* b^* = \sum_{i=1}^N \underbrace{\left((1+s^*)\lambda_i - s^* \mu_i \right)}_{=: x(s^*)_i} v_i.
\end{equation}
The last sum is an affine combination of the vertices since $\sum_i x(s^*)_i =1$. In other words, $x(s^*)$ is a feasible solution for the following minimisation problem:
\begin{equation}
 \RM(a) := \min \left\{ \|x\|_1 \; \bigg| \; x \in \R^N \,:\,  a = \sum_{i=1}^N x_i v_i \text{ and } 1 = \sum_{i=1}^N x_i\right\}.
\end{equation}
Moreover, the optimal value can be bounded as follows:
\begin{equation}
\label{eq:rom_ubound}
 \RM(a) \leq \sum_{i=1}^N |x(s^*)_i| \leq (1+s^*) \sum_{i=1}^N \lambda_i + s^* \sum_{i=1}^N \mu_i = 1 + 2s^* = 1 + 2 R(a).
\end{equation}

Assume $x^*$ is the optimal solution for $\RM(a)$. Then, we can rewrite $\RM(a)$, using $\sum_i x_i =1$, as follows:
\begin{equation}
 \RM(a) = \|x^*\|_1 = \sum_{i:\,x^*_i \geq 0} x^*_i - \sum_{i:\,x^*_i < 0} x^*_i = 1  + 2 s(x^*),  \quad \text{ with } s(x^*) := - \sum_{i:\,x^*_i < 0} x^*_i.
\end{equation}
Hence, the optimal affine combination for $a$ becomes
\begin{align}
  a &= \sum_{i:\,x^*_i \geq 0} x^*_i v_i - \sum_{i:\,x^*_i < 0} |x^*_i| v_i \\
    &= (1+s(x^*)) \underbrace{\sum_{i:\,x^*_i \geq 0} \frac{x^*_i}{1+s(x^*)} v_i}_{=:\beta^+} \; - \; s(x^*) \underbrace{\sum_{i:\,x^*_i < 0} \frac{|x^*_i|}{s(x^*)}}_{=:\beta^-} v_i.
\end{align}
Here,  the renormalised modulus of the affine coefficients form a convex combination and hence $\beta^\pm \in S$. Thus, we found a pseudo-mixture for $a$ and the parameter $s(x^*)$ can not be smaller than the total robustness of $a$:
\begin{equation}
 R(a) \leq s(x^*) \quad \Leftrightarrow \quad \RM(a) \geq 1 + 2 R(a).
\end{equation}
Combined with Eq.~\eqref{eq:rom_ubound}, this shows that the two measures are equivalent:
\begin{equation}
 \RM(a) = 1 + 2 R(a).
\end{equation}

Finally, let us remark that $\beta^{-}$ constructed from the optimal affine combination for $a$ is such that 
\begin{equation}
 R(a) = R(a||\beta^{-}).
\end{equation}

\section{On the dual RoM problem}
\label{sec:dual}

At this point, any analytical insight could be helpful in simplifying the problem. A standard method is dualising the problem. Clearly, by Slater's condition, strong duality holds and thus the dual problem is an equivalent definition for the Robustness of Magic. The dual problem is straightforwardly obtained as follows:
\begin{problem}[Dualised Robustness of Magic]
\label{prob:dual_rom}
Let  $\stab(n)=\{s_1,\dots,s_N\}$ be the set of stabiliser states.
Given a state $\rho$, solve the following problem:
\begin{align*}
  \textbf{max} \quad & \tr(\rho Y) \quad \text{over }Y \in H_D, \\
  \textbf{s.\,t.}\quad & | \tr(Y s_i) | \leq 1. 
 \end{align*}
\end{problem}
This formulation of the RoM has a particularly nice form. Thus, it seems at first that the dual problem might be easier to solve. Indeed, one can guess the following feasible solution:
\begin{equation}
 Y = \frac{1}{2^n}\sum_{i=1}^{4^n} \sgn\left(\tr(\rho w_i)\right) w_i.
\end{equation}
Here, $\{w_1,\dots,w_{4^n}\}$ denote the $n$-qubit Pauli operators which generate the $n$-qubit Pauli group $\Pa_n$.
Feasibility follows from the following calculation for a stabiliser state $s$ with stabiliser group $S<\Pa_n$:
\begin{equation}
\begin{split}
 |\tr(Ys)| &= \bigg|\frac{1}{2^n}\sum_{i=1}^{4^n} \sgn\left(\tr(\rho w_i)\right) \tr(w_i s)\bigg| \\
 &\leq \frac{1}{4^n}\sum_{i=1}^{4^n} \big|\sgn\left(\tr(\rho w_i)\right)\big| \sum_{g\in S} \big|\tr(w_i g)\big| \\
 &= \frac{1}{2^n}\sum_{i=1}^{4^n} ( \delta(w_i \in S) + \delta(-w_i \in S) ) \\
 &\leq \frac{1}{2^n} |S| = 1. 
\end{split}
\end{equation}
The corresponding objective value is
\begin{equation}
 \tr(\rho Y) = \frac{1}{2^n}\sum_{i=1}^{4^n} \sgn\left(\tr(\rho w_i)\right) \tr(\rho w_i) = \sum_{i=1}^{4^n} \left| \frac{\tr(\rho w_i)}{2^n}\right| = \|p(\rho)\|_1,
\end{equation}
where $p(\rho) \in \R^{D^2}$ is the coefficient vector of $\rho$ in the Pauli basis, i.e.~$p(\rho)_i = 2^{-n}\tr(\rho w_i)$. The objective value yields a lower bound to the RoM of $\rho$. Note that this bound, also called \emph{st-norm} $\|\rho\|_\mathrm{st}$, was already found in \cite{howard_application_2017} with different techniques and gives the following lower bound on the RoM of $\ket{H}^{\otimes n}$ and $\ket{T}^{\otimes n}$:
\begin{equation}
 1.207^n \leq \RM(\ket{H}^{\otimes n}), \qquad 1.366^n \leq \RM(\ket{T}^{\otimes n}).
\end{equation}

\section{Symmetries of 3-designs}
\label{sec:designs}

In this section, we characterise the symmery group associated with the projectors of certain \emph{$t$-designs}.

A \emph{complex projective $t$-design} is a finite family $(\psi_i)_{i=1}^N$ of unit vectors in $\C^d$ such that 
\begin{align}
\label{eq:t-design}
  \frac{1}{N} \sum_{i=1}^N \ket{\psi_i}\bra{\psi_i}^{\otimes t} = \frac{\Sym_{[t]}}{D_{[t]}},
\end{align}
where 
\begin{align}
  \Sym_{[t]}: \; (\C^d)^{\otimes t} \longrightarrow (\C^d)^{\otimes t}, \qquad
  \Sym_{[t]}:=
  \frac{1}{t!} \sum_{\pi \in S_t} \pi
  \label{eq:symtwirl}
\end{align}
is the orthogonal projection onto the totally symmetric subspace $\Sym((\C^d)^{\otimes t})$. Furthermore, $D_{[t]}=\binom{d+t-1}{t}$ is its dimension and $\pi\in S_t$ acts by permuting the factors of the tensor product $(\C^d)^{\otimes t}$.
Taking a partical trace of Eq.~\eqref{eq:t-design} shows that a $t$-design is also a $t-1$ design. 

As in the main part of this paper, we denote by $H_d$ the real vector space of Hermitian $d \times d$ matrices with the induced Hilbert-Schmidt inner product $(A,B):=\tr(AB)$. With respect to this inner product, we denote by $L^\dagger$ the adjoint of a linear map $L:\,H_d\rightarrow H_d$ and call $L$ orthogonal if it preserves the inner product, or equivalently, if $L^\dagger = L^{-1}$.

\begin{theorem}
 \label{thm:design_syms}
  Let $(\psi_i)_{i=1}^N \subset \C^d$ be a set of unit vectors.
 Let $L\in\End(H_d)$ be a linear map on Hermitian operators that permutes the projectors $(\ket{\psi_i}\bra{\psi_i})_{i=1}^N$. 
 \begin{enumerate}
 	\item
	  If $(\psi_i)_{i=1}^N$ is a 1-design, then $L$ is unital (i.e.\ $L(\one)=\one$).
 	\item
	  If $(\psi_i)_{i=1}^N$ is a 2-design, then $L$ is orthogonal and trace-preserving.
 	\item
	  If $(\psi_i)_{i=1}^N$ is a 3-design, then $L$ is of the form $L= U\cdot U^\dagger$, where $U$ is either a unitary or an antiunitary operator on $\C^d$.
\end{enumerate}
\end{theorem}

\begin{proof}
  Define $\rho_i:=\ket{\psi_i}\bra{\psi_i}$.

  \emph{1.}---Using $\Sym_{[1]}=\one$, we get
  \begin{equation}
   L(\one) = \frac{d}{N} \sum_{i=1}^N L(\rho_i) = \frac{d}{N} \sum_{i=1}^N \rho_i = \one,
  \end{equation}
  hence $L$ is unital.

  \emph{2.}---Let us define the traceless operators operators $f_i:=\rho_i - \one/d$. Using the fact that $\{\psi_i\}_i$ forms a 2-design, Eq.~(\ref{eq:symtwirl}), and the ``swap trick''
  \begin{align}\label{eq:swap trick}
	\tr(A B)=\tr(A\otimes B\,\pi)
  \end{align} 
  valid for the non-trivial element $\pi$ of $S_2$, one verifies the following for any traceless Hermitian operator $A\in H_d^0$:
  \begin{equation}
   \begin{split}
    \frac{1}{N} \sum_{i=1}^N \left(f_i, A\right)^2 &= \frac{1}{N} \sum_{i=1}^N \left[ \tr(\rho_i A) - \frac 1 d \tr(A) \right]^2 \\
    &= \frac{1}{N} \sum_{i=1}^N \tr(\rho_i^{\otimes 2} A^{\otimes 2}) \\
    &= \frac{1}{D_{[2]}} \tr(\Sym_{[2]} A^{\otimes 2}) \\
    &= \frac{1}{2D_{[2]}} \left( \tr(A)^2 + \tr(A^2) \right) \\
    &= \frac{\|A\|_2^2}{2D_{[2]}}.
   \end{split}
  \end{equation}
  In other words, the operators $(f_i)$ form a tight frame for the subspace $H_d^0\subset H_d$ of traceless Hermitian matrices.
  Moreover, setting $f_0 = c \one$ with $c^2=\frac{N(1-d)}{d^2+d^3}$, a similar calculation shows that the set $\{f_0, \dots, f_N\}$ forms a tight frame for all of $H_d$.
  
  By \emph{1.}, $L$ is unital and thus permutes the tight frame $\{f_0, \dots, f_N\}$. However, any map permuting the elements of a tight frame is orthogonal. Finally, orthogonal and unital maps preserve the trace:
 \begin{equation}
   \tr L(A) = \tr\one L(A) = \tr L^\dagger(\one) A = \tr L^{-1}(\one) A =  \tr \one A = \tr A.
 \end{equation}
 
 \emph{3.}---Consider the following trilinear function on $H_d$:
\begin{equation}
 F(A,B,C) := \frac{1}{N}\sum_{i=1}^N\tr(A\otimes B\otimes C\,\rho_i^{\otimes 3}).
\end{equation}
$F$ is invariant under $L$ since $L^\dagger = L^{-1}$ is also a symmetry of the projectors $\rho_i$:
\begin{equation}
\begin{split}
 F(L(A),L(B),L(C)) &= \frac{1}{N}\sum_{i=1}^N\tr\left( L(A)\otimes L(B)\otimes L(C)\,\rho_i^{\otimes 3} \right) \\
    &= \frac{1}{N}\sum_{i=1}^N\tr A\otimes B\otimes C\,\left(L^\dagger(\rho_i)\right)^{\otimes 3} \\
    &= \frac{1}{N}\sum_{i=1}^N\tr A\otimes B\otimes C\,\rho_i^{\otimes 3} \\
    &= F(A,B,C).
\end{split}
\end{equation}
We can explicitely evaluate $F$ by expanding $\Sym_{[3]}$ in terms of permutations and arguing as in Eq.~(\ref{eq:swap trick}).
This yields
\begin{equation}
\begin{split}
 F(A,B,C) &= \tr \left(A\otimes B\otimes C\,\frac{1}{N}\sum_{i=1}^N\rho_i^{\otimes 3} \right) \\
	&= \frac{1}{D_{[3]}} \tr \left(A\otimes B\otimes C\,\Sym_{[3]} \right) \\
    &=  \frac{1}{6 D_{[3]}} \Big( \tr(A)\tr(B)\tr(C) + \tr(A)\tr(BC) + \\ 
    & \qquad + \tr(AB)\tr(C) + \tr(AC)\tr(B) + \tr(ABC) + \tr(BAC) \Big).
\end{split}
\end{equation}
The first four terms are individually $L$-invariant since $L$ preserves the trace and is orthogonal. 
Hence, the $L$-invariance of $F$ implies
\begin{equation}
\begin{split}
  \tr(ABC) + \tr(BAC) 
  &= \tr L(A)L(B)L(C) + \tr L(B) L(A) L(C) \\
    &= \tr L^\dagger\big(L(A)L(B)\big)C + \tr L^\dagger\big( L(B) L(A) \big) C.
\end{split}
\end{equation}
Since this holds $\forall C\in H_d$, we get
\begin{align}
  AB + BA &= L^\dagger\big(L(A)L(B)\big) + L^\dagger\big( L(B) L(A) \big) \nonumber\\
 \Leftrightarrow L(AB + BA) &= L(A)L(B) + L(B)L(A) \nonumber\\
  \Leftrightarrow L(\{A,B\}) &= \{L(A),L(B)\}. \label{eq:jordan}
\end{align}
A linear automorphism on a matrix algebra fulfilling (\ref{eq:jordan}) is called a \emph{Jordan automorphism}. 
Our goal is to apply a known structure theorem that restricts that form of such maps \cite{herstein_1956}.
For the theorem to be applicable, we have to extend $L$ from a map on the real vector space of Hermitian matrices, to a map on the algebra $M_d(\C)$ of all matrices.
To this end, we use that every $A\in M_d(\C)$ can be written uniquely as $A = A_1 + i A_2$ where $A_{1,2}\in H_d$, and set
\begin{equation}
 \hat L(A) := L(A_1) + i L(A_2) \in M_d(\C).
\end{equation}
Clearly, this continuation yields a linear automorphism on $M_d(\C)$. Morover, since the anticommutator $\{\cdot,\cdot\}$ is bilinear, we get $\forall A,B\in M_d(\C)$:
\begin{equation}
 \hat L(\{A,B\}) = \{\hat L(A),\hat L(B)\},
\end{equation}
\ie~the continuation $\hat L$ to $M_d(\C)$ is a Jordan automorphism. 
It is also straightforward to check that orthogonality of $L$ implies that $\hat L$ is unitary with respect to the trace inner product.

It is known that every Jordan automorphism is either an algebra automorphism or algebra anti-automorphism \cite{herstein_1956}. Since every algebra automorphism is inner and $\hat L$ is unitary, $\hat L$ (and thus also $L\equiv\hat L|_{H_d}$) can in the first case be written as $\hat L=U\cdot U^\dagger$ for some $U\in U(d)$. In the second case, we can write $\hat L$ as a composition $\hat L=\hat L' \circ T$, where $\hat L'=U\cdot U^\dagger$ is an algebra automorphism and $T$ is the transposition map. For every Hermitian matrix, transposition coincides with complex conjugation as $A^T = (A^\dagger)^* = A^*$. Hence, we can write $L=U\mathcal{C}\cdot\mathcal{C}U^\dagger$, where $U\in U(d)$ and $\mathcal C$ is complex conjugation on $\C^d$. Hence, $L$ is in this case given by conjugation with the anti-unitary operator $U\mathcal{C}$.
\end{proof}

Since the qubit stabiliser state vectors in Hilbert space form a complex projective 3-design \cite{zhu_multiqubit_2017,kueng_qubit_2015,webb_clifford_2016}, we get the following corollary:

\stabsym*

\begin{proof}
  Theorem \ref{thm:design_syms} implies that every qubit stabiliser symmetry is given by conjugation with either an unitary or anti-unitary operator on the Hilbert space $\C^{2^n}$. 
 
  Theorem 2 in \cite{cormick_classicality_2006} implies that every unitary operator that preserves the set of stabiliser states is an element of the Clifford group, up to a global phase. 
 
 
 Furthermore, note that complex conjugation $\mathcal C$ preserves the set of stabiliser states. Thus, if $A$ is an anti-unitary operator preserving this set, $\mathcal{C}A$ is a perserving unitary operator. Hence, up to a phase, $\mathcal{C}A$ is Clifford and thus $A$ is anti-Clifford. Finally, this implies our claim that $\Aut(\SP_n)=\Ad(\ECl_n)$
\end{proof}

We note that the result is in general wrong for stabiliser states on odd-dimensional qudits.
This also means that the third conclusion of Thm.~\ref{thm:design_syms} is not in general true for 2-designs.
Concretely, take $(\psi_i)_i$ to be the set of stabiliser state vectors for $\C^d$, with $d$ a prime number larger than or equal to 5.
Then $(\psi_i)_{i=1}^N$ is a 2-design, but the group of linear symmetries of $\{\ket{\psi_i}\bra{\psi_i}\}_i$ contains maps that cannot be represented by a linear or anti-linear operator on $\C^d$.

\begin{proof}[Sketch of proof]
We sketch the proof of this claim in the language of \cite{gross_hudsons_2006}.
With each $a\in\Z_d^2$, one can associate a \emph{phase space point operator} $A(a)$. 
The $\{A(a)\}_a$ form a basis for $H_d$. 
The finite general linear group $GL(\Z_d^2)$ acts on this basis by permuting the indices $g\,A(a)=A(g\,a)$.
The expansion coefficients $W_\rho(a)$ of an operator $\rho$ with respect to the phase space point basis are the \emph{Wigner function} of the operator.
The stabiliser state $\rho_i=\ket{\psi_i}\bra{\psi_i}$ are exactly the set of Hermitian operators whose Wigner function is the indicator function of an affine line in $\Z_d^2$ \cite{gross_hudsons_2006}.
Clearly, the $GL(\Z_d^2)$-action introduced above preserves the set of affine lines and thus permutes the $\rho_i$.
As argued in the proof of Corollary~\ref{cor:stabsymmetries}, the group of (anti-)linear operators acting on the state vectors $\psi_i$ is the extended Clifford group $\ECl_n$.
To each $U$ in $\ECl_n$, one can associate a $g\in \Z_d^2$ such that $U A(a) U^{-1} = A(g\,a)$. 
But $g$'s that arise this way have determinant $\det g=1\mod d$ (if $U$ is unitary) or $\det g = -1\mod d$ (if $U$ is anti-unitary) \cite{appleby_symmetric_2005}. 
The claim follows,
as for $d\geq 5$, there are elements $g\in\GL(\Z_d^2)$ with determinant different from $\pm 1$.
\end{proof}

\section{Numerical implementation}
\label{sec:numerics}


Based on the discussion in Sec.~\ref{sec:id_symmetries}, we can construct a generic algorithm for generating projected stabiliser states by calling various oracles. $\Call{GraphRepresentatives}{n}$ generates suitable representatives of graph states. Here, these are given by \emph{connected} representatives of $\graph(n)/\sim_{LC,S_n}$ which were classified by \textcite{danielsen_classification_2006} up to 12 qubits and can by found in Ref.~\cite{danielsen_website}. $\Call{GeneratorMatrix}{G}$ computes the binary generator matrix of the graph state $\ket G$. Furthermore, $\Call{LocalSymplectic}{n,G}$ returns the set of local symplectic matrices, ideally up to the considered symmetry group. For the discussed cases in Sec.~\ref{sec:id_symmetries}, this is either the set of direct sums of $\{\one, \hat H, \hat H \hat S\}$ or $\{\one, \hat S\}$ up to the symmetry of the graph $G$. Finally, \Call{ProjectState}{$M',s$} and \Call{ProductState}{$v_1,\dots,v_k$} basically evaluate the weight enumerator formulas \eqref{eq:wt_enumerator} and \eqref{eq:wt_prod_states}.

\begin{algorithm}
 \caption{Algorithm for generating vertices of the projected stabiliser polytope}
 \label{alg:generation}
 \begin{algorithmic}
  \Require Maximum number of qubits $n_\mathrm{max}\geq 1$, set of vertices $\mathcal V_n$
  \For{$n=1,\dots,n_\mathrm{max}$}
    \For{$G \in \Call{GraphRepresentatives}{n}$}
        \State $M\gets \Call{GeneratorMatrix}{G}$
        \For{$S \in \Call{LocalSymplectic}{n,G}$}
            \State $M'\gets S \cdot M$
            \For{$s \in \{-1,1\}^{\times n}$}
                \State Add \Call{ProjectState}{$M',s$} to $\mathcal V_n$
            \EndFor
        \EndFor
    \EndFor
    \For{$(i_1,\dots,i_k) \in \Call{Partitions}{n}$}
        \For{$v_1 \in \mathcal V_{i_1},\dots, v_k\in \mathcal V_{i_k}$}
            \State Add \Call{ProductState}{$v_1,\dots,v_k$} to $\mathcal V_n$.
        \EndFor
    \EndFor
  \EndFor
 \end{algorithmic}
\end{algorithm}

Furthermore, we use a output-sensitive algorithm by \textcite{dula_new_1996} to compute the extremal points of the projected stabiliser polytope. This algorithm has time complexity $O(d N m)$ where $d$ is the dimension, $N$ the input size and $m$ the output size, \ie~the number of extremal points. It performs way better than a naive approach since the input size $N=O(2^{n^2})$ is much larger than the number of extremal points $m=O(2^n)$.

\section{Symmetry reduction of convex optimisation problems}
\label{sec:symmetry_red_app}

\subsection{Group projections}

The central tool for performing a symmetry reduction with respect to some (finite) group $G$ is the so-called \emph{$G$-projection}. Suppose $g$ is represented by $\rho:\,G\rightarrow \Gl(V)$ on a (real or complex) vector space $V$, we define a linear map $\Pi_G:\,V \rightarrow V$, the \emph{$G$-projection}, by
\begin{equation}
  \label{eq:G_twirl}
  \Pi_G := \frac{1}{|G|} \sum_{g\in G} \rho(g).
\end{equation}
The $G$-projection is well known in the representation theory of finite groups. In the physics literature, it is often called a \emph{twirl} or \emph{twirling operation}. Thus, we will also sometimes refer to it as $G$-twirl. 
For reference, we state some elementary properties of these maps without proof.

\begin{proposition}[Properties of $G$-projections]
\label{prop:G_twirl}
 Let $\Pi_G:\,V \rightarrow V$ be a $G$-projection. Then the following holds:
 \begin{enumerate}
  \item 
  $\Pi_G$ is a projection operator. 
  Its image is the subspace $V^G$ of fixed points of $G$. 
  \item If $V$ is an inner product space and $\rho$ is an orthogonal/unitary representation, then the projection is orthogonal/unitary.
  \item For all $x\in V$, it holds that
    \begin{equation}
     \Pi_G(x) = \frac{1}{|G\cdot x|} \sum_{y \in G\cdot x} y.
    \end{equation}
  \item $\Pi_G$ is constant on every orbit in $V/G$
  \item If $G=N \rtimes H$, then $\Pi_G = \Pi_N \circ \Pi_H =  \Pi_H  \circ \Pi_N$.
 \end{enumerate}
\end{proposition}


\subsection{Symmetries in convex optimisation}

A convex optimisation is the problem of minimising a convex function $F$ over a convex set $\mathcal{X}$. It can always be rewritten in standard form as follows: Let $F:\,\R^N\rightarrow\R$ be a convex function and $C:\,\R^N\rightarrow \R^K$ be a (generalised) convex function with respect to the component-wise partial order $\preceq$ on $\R^K$, \ie~every component of $C$ is convex. Furthermore, let $A:\,\R^N\rightarrow \R^M$ be an affine function. The problem is defined as \cite{boyd_2009}
\begin{equation}
\begin{split}
  \textbf{Minimise} \quad & F(x), \quad \text{for } x \in \R^N \\
  \textbf{subject to}\quad & A(x) = 0, \\
                    \quad & C(x) \preceq 0.
\end{split}
\label{eq:conv_opt}
\end{equation}
Here, the function $F$ is called the \emph{objective function} and the functions $C$ and $A$ are the (in-)equality \emph{constraints}. Depending on the convex set that is modelled, one distinguishes between many subclasses such as linear, conic or semi-definite programming.

We call $G$ a symmetry of the problem \eqref{eq:conv_opt}, if it acts on $\R^N$ such that the feasible set
\begin{equation}
 \mathcal{X} = \{ x \in \R^N \; | \; A(x) = 0, \; C(x) \preceq 0 \},
\end{equation}
and the objective function $F$ are left invariant. In particular, this will be the case if $G$ acts linearly on all vector spaces such that the objective function is $G$-invariant and the constraints are $G$-equivariant, \ie~for all $x\in \R^N$ and $g\in G$ it holds
\begin{equation}
\begin{split}
 F(g\cdot x) &= F(x), \\
 A(g\cdot x) &= g \cdot A(x), \\
 C(g\cdot x) &= g \cdot C(x).
\end{split}
\label{eq:conv_opt_equivariance}
\end{equation}
Again, note that the $G$-action is different on the left and right hand side. Additionally, for $G$ to be a proper symmetry, we require that its representation on $\R^K$ is given by \emph{order automorphisms}, \ie~
\begin{equation}
 p \preceq q \quad \Longleftrightarrow \quad g\cdot p \preceq g\cdot q \quad \forall p,q\in\R^K,\;g\in G
\end{equation}
Consequently, both the inequality $C(g\cdot x)\preceq 0$ and the equality constraint $A(g\cdot x)=0$ are fulfilled if and only if they hold for $x$. Hence, $x\in \R^N$ is a feasible solution of Eq.~\eqref{eq:conv_opt} iff its orbit is feasible. Moreover, the objective function is constant on every orbit and thus any optimal solution $x^*$ will have an orbit of optimal solutions. 

The key point for the simplification of the problem is that all functions are \emph{convex} ($A$ is even affine). Let us again slightly abuse notation and denote with
\begin{equation}
 \Pi_G = \frac{1}{|G|} \sum_{g\in G} g,
\end{equation}
all $G$-projections on the respective spaces. Using this, we will derive two important consequences of $G$-equivariance of $A$ and $C$. First, we evaluate the affine function $A$:
\begin{equation}
\label{eq:APiG}
  A\circ\Pi_G(x) = \frac{1}{|G|} \sum_{g\in G} A\circ g(x) = \frac{1}{|G|} \sum_{g\in G} g \circ A(x) = \Pi_G \circ A(x).
\end{equation}
Recall that $C$ is convex w.r.t.~to the component-wise order $\preceq$ and that every $g\in G$ preserves this order. Thus, $\Pi_G$ preserves order, too, and it follows:
\begin{equation}
\label{eq:CPiG}
  C\circ\Pi_G(x) \preceq \frac{1}{|G|} \sum_{g\in G} C\circ g(x) = \frac{1}{|G|} \sum_{g\in G} g \circ C(x) \preceq \Pi_G \circ C(x).
\end{equation}
Suppose $x$ is a feasible solution, then by these relations, its $G$-projection $x^G=\Pi_G(x)$ is feasible, too. Following the same argument as above, we get $F(x^G)\leq F(x)$. Finally, we find the following results:

\begin{lemma}
 Every $G$-symmetric convex optimisation problem has a $G$-invariant optimal solution.
\end{lemma}

\begin{proof}
 Be $x^*$ a optimal solution, then $\Pi_G(x^*)$ is $G$-invariant, feasible and $F(\Pi_G(x^*)) \leq F(x^*)$. Hence, $\Pi_G(x^*)$ is optimal, too.
\end{proof}

\begin{theorem}[Symmetry reduction of convex optimisation problems]
\label{thm:sym_red_co}
 The convex optimisation problem \eqref{eq:conv_opt} with symmetry group $G$ is equivalent to the following, \emph{symmetry-reduced} convex optimisation problem: 
 \begin{equation}
    \begin{split}
    \textbf{Minimise} \quad & F^G(x), \quad \text{for } x \in X^G \\
    \textbf{subject to}\quad & A^G(x) = 0, \\
                        \quad & C^G(x) \preceq 0.
    \end{split}
    \label{eq:conv_opt_red}
 \end{equation}
 With $F^G:\,X^G\rightarrow \R$, $A^G:\,X^G\rightarrow Y^G$ and $C^G:\,X^G\rightarrow Z^G$ being functions such that
 \begin{equation}
 \label{eq:red_funcs}
  F^G\circ \Pi_G = F, \qquad A^G\circ \Pi_G = \Pi_G \circ A, \qquad C^G\circ \Pi_G = \Pi_G \circ C,
 \end{equation}
 and $X^G$, $Y^G$, $Z^G$ being the $G$-invariant subspace of $X=\R^N$, $Y=\R^M$ and $Z=\R^K$.
\end{theorem}

\begin{proof}
 First, it should be clear that the functions $F^G$, $A^G$ and $C^G$ exist and are well-defined by Eq.~\eqref{eq:red_funcs}. Moreover, we compute for $x,y\in X^G$ and $t\in[0,1],s\in\R$:
 \begin{equation}
  \begin{split}
   F^G( t x + (1-t) y ) &= F( t x + (1-t) y ) \\ 
        &\leq t F(x) + (1-t) F(y) \\
        &= t F^G(x) + (1-t) F^G(y), \\
   A^G( s x + (1-s) y ) &= \Pi_G\circ A( s x + (1-s) y ) \\ 
        &=  s \Pi_G \circ A(x) + (1-s) \Pi_G\circ A(y) \\
        &= s A^G(x) + (1-s) A^G(y), \\
   C^G( t x + (1-t) y ) &= \Pi_G\circ C( t x + (1-t) y ) \\
        &\preceq  t \Pi_G \circ C(x) + (1-t) \Pi_G\circ C(y) \\
        &= t C^G(x) + (1-t) C^G(y).
  \end{split}
 \end{equation}
 Hence, $F^G$ and $C^G$ are convex and $A^G$ as an affine function.

 Suppose $x\in \R^N$ is a feasible solution of the original problem \eqref{eq:conv_opt} which can be assumed to be $G$-invariant, \ie~$x\in X^G$. It will be feasible for the reduced problem since
 \begin{equation}
  \begin{split}
   A^G(x) &= A^G\circ\Pi_G(x) = \Pi_G\circ A(x) = 0, \\
   C^G(x) &= C^G\circ\Pi_G(x) = \Pi_G\circ C(x) \preceq 0,
  \end{split}
 \end{equation} 
 and $F^G(x)=F(x)$. 
 
 Next, suppose $x^G$ is feasible for the reduced problem. By the same line of argumentation we get due to Eq.~\eqref{eq:APiG}:
 \begin{equation}
   0 = A^G(x^G) = A^G\circ\Pi_G(x) = \Pi_G\circ A(x) = A\circ\Pi_G(x) = A(x).
 \end{equation} 
 In the same fashion, we compute using Eq.~\eqref{eq:CPiG}:
 \begin{equation}
   0 \succeq C^G(x^G) = C^G\circ\Pi_G(x) = \Pi_G\circ C(x) \succeq C\circ\Pi_G(x) = C(x).
 \end{equation} 
 Hence, $x^G$ is feasible for the original problem and $F^G(x^G)=F^G\circ\Pi_G(x^G) = F(x^G)$.
 
 Finally, this implies that the optimal objective values have to agree: Suppose $x^*$ and $x^G_*$ are ($G$-invariant) optimal solutions for the original and the reduced problem, respectively. Then, $F^G(x^*)=F(x^*)$ and $F(x^G_*)=F^G(x^G_*)$. But since both $x^*$ and $x^G_*$ are feasible for both problems, $F(x^G_*)\neq F(x^*)$ would be a contradiction to the optimality of the solutions.

\end{proof}

\subsection{Affine constraints and symmetries}

In the remainder of this work, both $A$ and $C$ will be affine maps and originate from a set of points $\mathcal{V}$ that span a polytope $\mathcal{P}$. The symmetry group $G$ leaves $\mathcal{P}$ invariant and hence introduces permutations on $\mathcal{V}$. This will lead naturally lead to $G$-equivariance of these functions, as we will see in the following.

To simplify the discussion, we will focus on the function $A$. We can write the affine function $A$ as
\begin{equation}
 A(x) = \sum_{i=1}^N x_i v_i + v_0
\end{equation}
Here, $\mathcal{V}:=\{v_1,\dots,v_N\}\subset Y$ are the column vectors of the matrix representing the linear part of $A$ and $v_0$ is its affine part. Suppose $G$ is represented on $Y$ such that it leaves the set $\mathcal{V}$ invariant and fixes $v_0$\footnote{In general, $G$-equivariance requires that the action preserves the range of $A$ which is a weaker condition.}. Hence, it can by identified with the left action of some subgroup of the symmetric group $S_N$ on the index set $[N]=\{1,\dots,N\}$ via $g\cdot y_i =: y_{\pi_g(i)}$ for some $\pi_g\in S_N$. We can associate a right action on $X$ with this left action by $(x\cdot g)_i := x_{\pi_g^{-1}(i)}$. This action is clearly linear and such that for all $g\in G$:
\begin{equation}
 \sum_{i=1}^N x_i\, (g\cdot v_i) = \sum_{i=1}^N x_i \, v_{\pi_g(i)} = \sum_{i=1}^N x_{\pi_g^{-1}(i)}\, v_i = \sum_{i=1}^N (x\cdot g)_i\, v_i.
\end{equation}
In particular, the function $A$ is $G$-equivariant:
\begin{equation}
 g\cdot A(x) = \sum_{i=1}^N x_i\,(g\cdot v_i) + g\cdot v_0 =\sum_{i=1}^N (x\cdot g)_i\, v_i + v_0 = A(x\cdot g).
\end{equation}
To make use of Thm.~\ref{thm:sym_red_co}, we have to compute the function $A^G$. Note that $\Pi_G$ is constant on the every orbit $O\in [N]/G$ and hence $\Pi_G(v_j)=: w_O$ for all $j \in O$:
\begin{equation}
\begin{split}
 \Pi_G\circ A(x) &= \sum_{i=1}^N x_i \, \Pi_G(v_i) + v_0 \\
            &= \sum_{O\in [N]/G}\sum_{j\in O} x_j \, \Pi_G(v_j) + v_0 \\
            &= \sum_{O\in [N]/G} \left(\sum_{j\in O} x_j\right) w_O + v_0 \\
            &= \sum_{O\in [N]/G} y_O\, w_O + v_0,
\end{split}
\end{equation}
where in the last step we set $y_O = \sum_{j\in O} x_j$. Finally, we have to turn this into a map on $X^G$. Note that the right permutation action of $G$ on $X=\R^N$ partitions the standard basis $\{e_1,\dots,e_N\}$ into $L$ orbits $O_1,\dots,O_L$ corresponding to $[N]/G$. Next, the linear spans $X_j=\langle O_j \rangle$ of these orbits provide a decomposition of $X=\bigoplus_j X_j$ and $G$ acts transitively on every orbit. Hence, $\Pi_G(X_j)$ is one-dimensional and $\Pi_G(X)=\bigoplus_j \Pi_G(X_j)$ due to linearity. This implies that $\dim X^G = L = |[N]/G|$. Hence, the $y_O$ are the components of a vector $y\in X^G$ w.r.t.~the basis $\tilde{e}_O = \sum_{j\in O}e_j$. Note that if we normalise that basis as $e_O = \frac{1}{|O|}\tilde{e}_O$, then the new components are $\bar x_O = \frac{1}{|O|}\,y_O$, which are exactly the components of $\Pi_G(x)$. Hence, the induced map on $X^G$ is
\begin{equation}
 A^G(x) = \sum_{j=1}^L x_j w_j + v_0.
\end{equation}

As stated in the beginning of this subsection, the points $\mathcal{V}$ are the extremal points of a polytope $\Po$ and $G$ is as subgroup of the polytope symmetries $\Aut(\Po)$. We saw that the symmetry reduction corresponds to projecting the vertices of the polytope, and hence the polytope itself, onto the $G$-invariant subspace. This is equivalent to taking its intersection with this subspace as the following lemma states:

\begin{lemma}[Projection with Polytope Symmetries]
\label{lem:poly_proj}
 Be $G < \Aut(\Po)$ a subgroup. Then, the $G$-projection of $\Po$ is contained in $\Po$, $\Pi_G(\Po) \subset \Po$. More precisely, $\Pi_G(\Po) = \Po \cap X^G$.
\end{lemma}

\begin{proof}
 For all $x\in\Po$, we have $G\cdot x\subset \mathcal{P}$ and $\Pi_G(x)$ is a convex combination of points in $\Po$, hence in $\Po$ itself. Moreover, it holds $\Po \cap X^G = \Pi_G(\Po \cap X^G) \subset \Pi_G(\Po)$. The converse direction follows since  $\Pi_G(\Po) \subset X^G$ and $\Pi_G(\Po) \subset \Po$, thus $\Pi_G(\Po)\subset \Po \cap X^G$, which shows $\Pi_G(\Po) = \Po \cap X^G$.
\end{proof}

Finally, we want to remark that for computing the projection of the vertices $\{v_1,\dots,v_M\}$, it is sufficient to compute $\Pi_G(w_O)$ for some representatives $w_O$ of the orbits $O\in \mathcal{V}/G$ since the projection only depends on the orbit.

\printbibliography

\end{document}